\documentclass[aps,prl,twocolumn,groupedaddress,amsmath,superscriptaddress]{revtex4-2}

\usepackage{amsfonts}
\usepackage[dvips]{graphicx}
\usepackage{color}
\usepackage{booktabs}
\usepackage{graphicx}
\usepackage{amssymb}
\usepackage{hyperref}
\usepackage{color}
\usepackage[caption=false]{subfig}
\usepackage{mathrsfs}
\usepackage{cleveref}
\usepackage{lipsum}
\usepackage{isomath}
\usepackage{amsmath} % 数学符号支持
\usepackage{mathtools} % provides \mathclap for aligned annotations
\usepackage{amsfonts} % 支持 \mathbb
\usepackage{algorithm} % 算法环境
\usepackage{algorithmic} % 算法命令支持
\usepackage{amsthm}
\usepackage{epstopdf}
\usepackage{txfonts}
\usepackage{dsfont}
\usepackage{placeins}
\usepackage{setspace}
\usepackage{multirow}
\usepackage{tikz}
\usetikzlibrary{arrows.meta,shapes.geometric}
\usepackage{tabularx}
\usepackage[T1]{fontenc}
\usepackage{colortbl}
\usepackage{hhline}
\usepackage{outlines}
\usepackage{xcolor}
\usepackage[bottom]{footmisc}
\usepackage[normalem]{ulem}
\allowdisplaybreaks[4]

\makeatletter
\definecolor{nblue}{rgb}{0.3,0.3,1.0}%229
\definecolor{ngreen}{rgb}{0.2,0.7,0.2}%161
\definecolor{nred}{rgb}{0.9,0.1,0}%711&900
\definecolor{nblack}{rgb}{0,0,0}
\definecolor{nyellow}{rgb}{1.0,0.75,0.0}

\newcommand{\beq}{\begin{equation}}
\newcommand{\eeq}{\end{equation}}

\newtheorem{proposition}{Proposition}

\newtheorem{lemma}{Lemma}
\newtheorem{corollary}{Corollary}

\newtheorem{theorem}{Theorem}

\makeatletter

\usepackage{babel}

\begin{document}
 
\title{Bounding Fidelity of Non-Stabilizer States from Local Pauli Measurements}
\author{Yi Li}
\affiliation{State Key Laboratory of Artificial Microstructure and Mesoscopic Physics, School of Physics, Peking University, Beijing 100871, China}
\affiliation{Naturwissenschaftlich-Technische Fakult\"{a}t, Universit\"{a}t Siegen, Walter-Flex-Stra{\ss}e 3, 57068 Siegen, Germany}
\author{Otfried G\"uhne}
\affiliation{Naturwissenschaftlich-Technische Fakult\"{a}t, Universit\"{a}t Siegen, Walter-Flex-Stra{\ss}e 3, 57068 Siegen, Germany}
\author{Qiongyi He}
\email{qiongyihe@pku.edu.cn}
\affiliation{State Key Laboratory of Artificial Microstructure and Mesoscopic Physics, School of Physics, Peking University, Beijing 100871, China}
\affiliation{Frontiers Science Center for Nano-optoelectronics \& Collaborative Innovation Center of Quantum Matter, Peking University, Beijing 100871, China}
\affiliation{Collaborative Innovation Center of Extreme Optics, Shanxi University, Taiyuan 030006, China}
\affiliation{Hefei National Laboratory, Hefei 230088, China}

\begin{abstract}
Nonstabilizerness, commonly referred to as magic, is essential for universal quantum computation. 
Yet, for pure $n$-qubit states this property is directly linked 
to the hardness of direct fidelity estimation. 
We propose a scheme to bound the fidelity on non-stabilizer states by choosing operators with low measurement overhead and using parallel local Pauli measurements to estimate compatible observables. 
We selectively measure operators that share a permutational symmetry without imposing any assumption on the physical system. 
This selection reduces the optimization to a linear program with only $O(n)$ variables for bounding fidelity. 
Our methods can be extended to bound nonlinear quantities
like the purity via convex quadratic programs. 
For hypergraph states that can exhibit extensive nonstabilizerness, we demonstrate for up to $n=500$ qubits that our framework substantially reduces the measurement overhead under both local and global depolarizing noise. 
These results show that symmetry induced by measurement selection reduces both measurement and computational costs, enabling scalable characterization of magic resources in quantum computing and quantum networks.

\end{abstract}

\maketitle

%%%%%%%%%%%%%%%%%%%%%%%%%%%%%%
\textit{Introduction---}
Multipartite entangled states are central resources for quantum networks, distributed sensing, and quantum computing~\cite{rmp_2023_networks,nature_2008_quantumnetworks,science_2018_quantumnetworks3, prl_2018_sensing,qst_2021_review_sensing,prl_2001_graphstate,nature_2024_Luckin}. Recent experimental advances have enabled the preparation and distribution of such states at increasingly large scales~\cite{nature_2023_51qubits,nature_2024_60atom,arxiv_2025_120qubits,np_2026_zuchongzhi,nc_2024_haohuawang}. 
Alongside this progress, reliable state characterization is indispensable for validating their experimental realization~\cite{nrp_2020_review}. 
The hardness of quantum state tomography has shifted attention toward the direct estimation of task-relevant quantities~\cite{ieee_2017_hardness_qst}. 
Among these quantities, fidelity to a target state serves as a central figure of merit for benchmarking state preparation and witnessing genuine multipartite entanglement (GME)~\cite{nature_2023_51qubits,nature_2024_60atom,arxiv_2025_120qubits,np_2026_zuchongzhi,nc_2024_haohuawang,prl_2005_GMEwitness,prl_2026_learningmixedstates,prl_2024_versatileFid}. 
Existing approaches to its estimation include shadow tomography~\cite{shadow_2018_aaronson,shadow_2020_hyh,NRP_2022_Huang,prl_2022_optimizedshadow} and direct fidelity estimation~(DFE)~\cite{prl_2011_directfidelity,prl_2011_dfe2}. 
These schemes are particularly efficient for stabilizer states, which are characterized as eigenstates of a certain set of tensor products of Pauli matrices. 
This simple description gives a natural choice of the measurements for fidelity estimation, but it also implies that stabilizer states have low complexity and can be efficiently simulated classically~\cite{pra_2004_DG_SA}.

Going beyond, non-stabilizer states provide magic required for universal quantum computation~\cite{njp_2014_nonsta_advantage,pra_2005_nonsta_advantage,arXiv_1998_DG_nonstab_advantage,pra_2004_DG_SA,nature_2017_review_nonstab_advantage}, yet the nonstabilizerness is closely linked to the difficulty of DFE based on local Pauli measurements~\cite{pra_2023_Leone}. Meanwhile, sample-efficient schemes typically rely on nonlocal measurements~\cite{shadow_2020_hyh,prl_2026_huangjun_fidelity,arXiv_2025_huangjunzhu_fidelity,NRP_2023_randomizedmeasurement}. 
These challenges have been exemplified by hypergraph states (also known as locally maximally entangleable states), which form a representative family of magic resources~\cite{prapplied_2019_hypergraph,pra_2009_LME,JPA_2014_Otfried,njp_2013_marcushuber,PRXQuantum_2022_ziwenliu}. 
They generalize graph states including the well-known Greenberger-Horne-Zeilinger states and cluster states~\cite{pra_2004_graphstate}. 
Unlike graph states that are stabilizer states and
have no magic, or many random states that have nearly maximal magic but cannot provide nontrivial speedups over classical computers~\cite{PRXQuantum_2022_ziwenliu}, hypergraph states can exhibit extensive magic and serve as resources for quantum error correction~\cite{JPA_2017_hypergraph_QECC} and Pauli-universal measurement-based quantum computation~\cite{sr_2019_hypergraph,pra_2019_hypergraph_otfried,npj_2016_hypergraph_MBQC}. 
They also occur in quantum many-body physics~\cite{npj_2016_hypergraph_MBQC,pra_2018_spto,PRB_2012_hypergraph_SPTO}. 
% However, they generally have broader Pauli support than stabilizer states, making DFE costly~\cite{prapplied_2019_hypergraph}. 

\begin{figure}[t]
\vspace{-0.2cm}
    \centering
\includegraphics[width=1\linewidth]{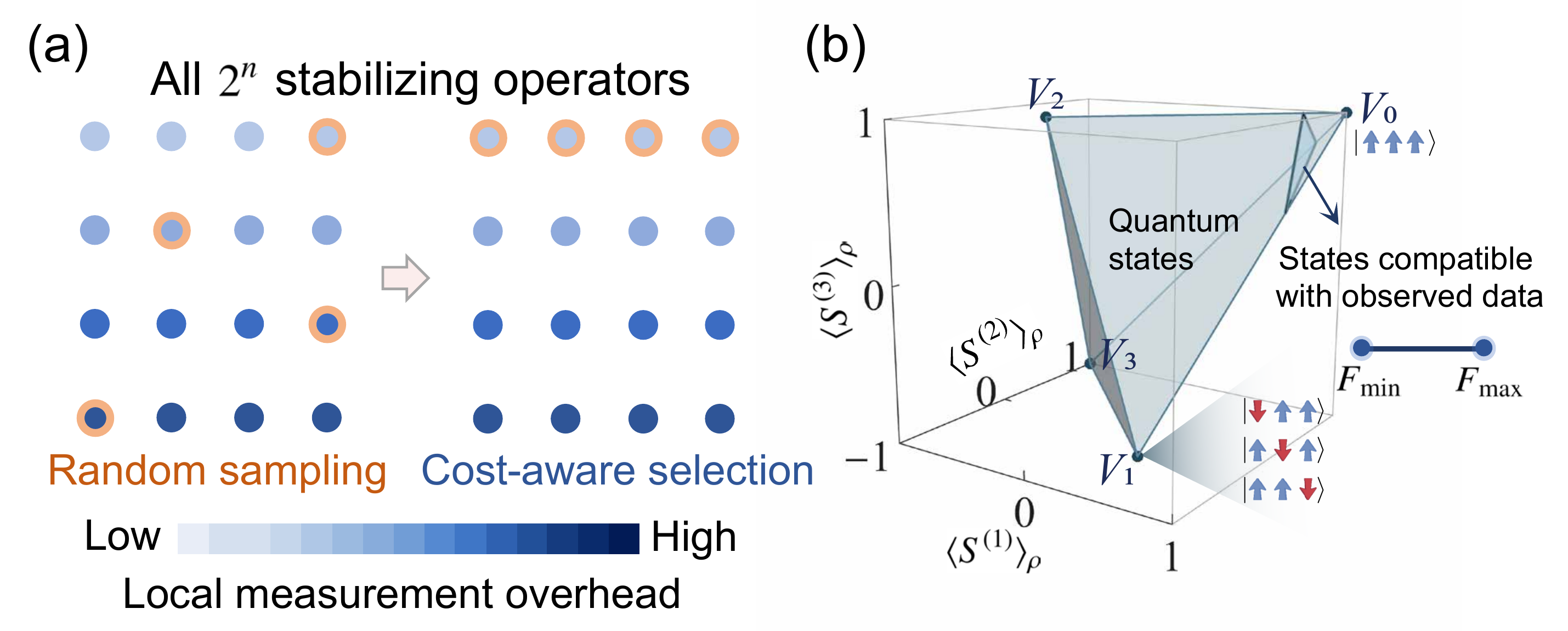}
    \caption{\label{fig:framework}
    Measurement selection and a geometric fidelity-bounding framework. 
    (a)~The fidelity to any $n$-qubit pure target state can be estimated by measuring operators randomly selected from its nonlocal $2^n$ stabilizing operators. Since the measurement overhead under parallel local Pauli measurements can vary widely across these operators, we select operators with low measurement overhead. 
    (b)~Joint numerical range of the stabilizer moments $\{S^{(1)},S^{(2)},S^{(3)}\}$ defined in Eq.~\eqref{eq:fidelity_spinmoment}, forming a tetrahedron with all quantum states mapped inside for $n = 3$. 
    All $2^3$ common eigenstates illustrated in the spin picture are mapped to the four vertices, and the vertex $V_0$ can only be achieved by the target state. 
    Suppose $\langle S^{(1)}\rangle = 0.8$ is observed, all compatible states map to the corresponding slice, from which the fidelity interval $[F_{\min},F_{\max}]$ can be obtained.} 
    \vspace{-0.5cm}
\end{figure}

The heavy measurement overhead motivates bounding fidelity from limited data, as, from a practical perspective, many certification tasks require only an informative fidelity bound rather than an exact estimation. This naturally relates to a fundamental question~\cite{pra_1999_fakeentanglement,njp_2006_plenio,prl_2007_otfried,njp_2007_jens,prl_2004_maciej,pra_2020_BoundingFidelity_matteo,pra_2019_Kalev_BoundingFidelity,pra_2015_boundingfidelity,pra_2006_boundingfidelity,jmo_2009_boundingfidelity}: \textit{Given a desired property, what is the least value it can take among all quantum states compatible with the available measurement data?} 
Such a task has been formulated as a minimization over all physical states consistent with the observed data, allowing all unmeasured degrees of freedom to vary in the least favorable way. % permitted by quantum mechanics 
Nonetheless, the underlying state space generally grows exponentially with system size. 
% and analytic solutions are available only in specific cases~\cite{pra_2019_Kalev_BoundingFidelity,pra_2015_boundingfidelity,pra_2006_boundingfidelity}. 
Similar optimizations and the associated numerical bottlenecks also arise broadly in quantum cryptography~\cite{rmp_2024_armin,np_2026_maximalrandomness}.

In this work, we develop a scalable framework for bounding fidelity to non-stabilizer states from local Pauli measurement data. 
The key idea comes from the fact that {\it every} $n$-qubit pure state is the unique common eigenstate of $2^n$ 
{\it nonlocal} stabilizing operators~\cite{pra_2021_generalizedstabilizerstates,prl_2019_sciara,pra_2023_Leone}. 
As shown in Fig.~\ref{fig:framework}~(a), the measurement overhead varies widely across these operators for certain target states. 
Rather than randomly sampling them in the spirit of DFE and selecting operators that may require exponentially many Pauli measurements, our protocol bounds the fidelity by selectively measuring operators with low measurement overhead. 
The selected set can be systematically enlarged to further tighten the bounds, while the associated state space projection described by the joint numerical range (JNR)~\cite{prl_2020_numericalrange,JNR_2013_karol} remains a simple polytope, see Fig.~\ref{fig:framework}~(b).

In this way, our protocol not only reduces the measurement cost but also turns the worst-case optimization into a linear program~(LP) with only $(n+1)$ variables. 
Our approach also extends to bounding nonlinear quantities such as purity, without increasing the number of optimization variables. 
This framework is well suited to hypergraph states, for which our framework significantly reduces the sample complexity compared to the standard DFE. We demonstrate this advantage under both global and local depolarizing noise for up to $n=500$ qubits. 
Moreover, this framework allows the same data to simultaneously bound fidelities of different states \cite{prl_2024_otfried}. 
Finally, our formulation easily extends existing analytic fidelity bounds, thereby strengthening fidelity-based GME witnesses~\cite{prl_2005_GMEwitness}.

\textit{Framework for bounding fidelity using stabilizing operators---}Consider an $n$-qubit pure state $|\psi\rangle$ in the Hilbert space $\mathcal{H}=(\mathbb{C}^2)^{\otimes n}$ of dimension $d=2^n$, and let $[n]\coloneqq\{1,\ldots,n\}$. 
Any pure state can be prepared by a unitary quantum circuit $U$ as $|\psi\rangle=U|0\rangle^{\otimes n}$, which is called a \textit{non-stabilizer state} if $U$ cannot be chosen as a Clifford circuit. 
In this work, we define the generators $g_i\coloneqq UZ_iU^\dagger$ for $i\in[n]$, where $Z_i$ denotes the Pauli $Z$ matrix on the $i$-th qubit. 
These generators are mutually commuting, Hermitian, unitary, traceless, and independent, but generally nonlocal. 
By construction, $|\psi\rangle$ is their unique common $+1$ eigenstate~\cite{pra_2021_generalizedstabilizerstates,prl_2019_sciara,pra_2023_Leone}.

These generators generate an abelian group $\mathcal{S} \coloneqq  \{s_{\boldsymbol{x}} = \prod_{i = 1}^n g_i^{x_i}\}_{\boldsymbol{x} \in \mathbb{Z}_2^n }$, whose elements are called stabilizing operators. 
The state is also the unique common $+1$ eigenstate of this group, and can therefore be written as $|\psi \rangle\langle \psi | = \prod_{i\in[n]} {\mathbb{I} + g_i \over 2}$, which expands into the uniform sum over all products of the generators. 
Grouping these products by their order, we define the \textit{order-$k$ stabilizer moment} and obtain the following expansion of the target-state projector:
\begin{equation}
\label{eq:fidelity_spinmoment}
    \begin{aligned}
        S^{(k)} \coloneqq {1\over L_k}\sum_{I\subseteq[n],|I| = k} \prod_{i\in I} g_i, \quad
        |\psi \rangle \langle \psi| = {1\over 2^n}\left(\mathbb{I} +\sum_{k = 1}^n L_kS^{(k)}\right),
    \end{aligned}
\end{equation}
where $L_k=\binom{n}{k}$ and $k\in[n]$. 
%For the stabilizer states, the stabilizer set $G$ consists of $n$-qubit Pauli operators. 
The stabilizing operators admit an orthonormal eigenbasis $\{|\boldsymbol{\lambda}\rangle\}_{\boldsymbol{\lambda}\in\mathbb{Z}_2^n}$, labeled by bit strings $\boldsymbol{\lambda}=(\lambda_1,\cdots,\lambda_n)$ such that $g_i|\boldsymbol{\lambda}\rangle = (-1)^{\lambda_i}|\boldsymbol{\lambda}\rangle$. 
Here, $\lambda_i$ labels the generator $g_i$ rather than a physical qubit, so $|\boldsymbol{\lambda}\rangle$ need not be a product state and can be highly entangled.

To estimate the fidelity between a laboratory-prepared state $\sigma$ and a target state $|\psi\rangle$, Eq.~\eqref{eq:fidelity_spinmoment} suggests a direct route by randomly sampling elements in $\mathcal{S}$ as in DFE~\cite{prl_2011_directfidelity,pra_2023_Leone}. 
In fact, for several families of target states, different stabilizing operators have different Pauli measurement overheads. 
The Pauli expansion of each stabilizing operator $s_{\boldsymbol{x}}$ contributing to $S^{(k)}$ contains a number of terms bounded solely by its Hamming weight $k\equiv|\boldsymbol{x}|$. 
As detailed in Appendix~\ref{app:one-step-tfim}, this number is at most: (1) $O(3^k)$ terms for Clifford-enhanced product states~\cite{arxiv_2026_sater_cps_certification,PRXQuantum_2025_cps}; (2) $2^{O(k)}$ terms for shallow circuits composed of arbitrary two-qubit gates~\cite{stoc_2024_huang_shallow_circuits}; and (3) $O(2^{k})$ terms for one-step Trotterized transverse-field Ising states~\cite{prl_2026_learningmixedstates,nature_2023_kim}. 
When the parallel measurements are taken into account, the number of Pauli settings required to measure $s_{\boldsymbol{x}}$ is also bounded by: (1) $O(2^{k-1})$ for hypergraph states~\cite{njp_2013_marcushuber,JPA_2014_Otfried}; and (2) $O(2^{k})$ for weighted graph states~\cite{review_2006_graphstate,prl_2006_weightedgraph}. 
These bounds are given in the worst case, which can be further reduced when the structure of the target state is specified.

Therefore, we propose to bound the fidelity by selectively measuring stabilizing operators with low measurement overhead. 
This naturally leads to a fundamental problem~\cite{pra_1999_fakeentanglement,njp_2006_plenio,jmo_2009_boundingfidelity,prl_2007_otfried,njp_2007_jens,prl_2004_maciej,pra_2020_BoundingFidelity_matteo,pra_2019_Kalev_BoundingFidelity,pra_2015_boundingfidelity,pra_2006_boundingfidelity}: % of determining the least fidelity compatible with limited data
finding the least fidelity consistent with observed expectation values $\{ \langle s_{\boldsymbol{x}}\rangle_\sigma \}_{\boldsymbol{x}\in\mathcal{M}}$ for a chosen index set $\mathcal{M}$, %the problem has been formulated as follows
\begin{equation}
\label{eq:boundingfidelity}
    F_{\min}(\mathcal{M})
    =
    \min_{\rho\in\mathcal{D}(\mathcal{H})}
    \left\{
    \langle\psi|\rho|\psi\rangle:
    {\rm Tr}(\rho s_{\boldsymbol{x}})
    =
    \langle s_{\boldsymbol{x}}\rangle_\sigma,\ 
    \boldsymbol{x}\in\mathcal{M}
    \right\}.
\end{equation}
Here, $\mathcal{D}(\mathcal{H})$ denotes the set of all $n$-qubit quantum states, namely positive semidefinite operators with unit trace. 
Such an optimization is a semidefinite program (SDP) over all quantum states, which can be reduced to an LP over $2^n$ variables~\cite{rmp_2024_armin,arxic_2025_lisa} (see also Appendix~\ref{app:A}). 
% both the objective function and the constraints are invariant under the action $\rho \mapsto s \rho s^{\dagger}$ for all $s \in \mathcal{S}$. 
% Consequently, we can assume $\rho$ is diagonal in the common eigenbasis $\{ |\boldsymbol{\lambda}\rangle \}$ without loss of generality, yielding an LP over its $2^n$ diagonal entries (see Appendix~\ref{app:A}). 
However, this LP remains exponentially large and analytical solutions are known only for specific choices of measured operators or target states~\cite{jmo_2009_boundingfidelity,pra_2006_boundingfidelity,pra_2015_boundingfidelity,pra_2019_Kalev_BoundingFidelity}. 
In what follows, we show that the %order-$k$ 
stabilizer moments play two central roles in our protocol. They guide the selection of stabilizing operators with low measurement overhead and thus reduce the measurement cost, while their simple geometry in the projected state space reduces the exponentially large LP to a linear-size LP.

% \textit{Joint Numerical range of the stabilizing operators---}
\textit{Geometric characterization---}The JNR of $\kappa$ Hermitian operators $O_1,O_2,\cdots,O_\kappa$ acting on the Hilbert space $\mathcal{H}$ is defined by \cite{prl_2020_numericalrange,JNR_2013_karol}
    \begin{equation}
        W(O_1,O_2,\cdots,O_\kappa) \coloneqq \{(\langle O_1 \rangle_\rho, \cdots, \langle O_\kappa \rangle_\rho)^{\mathsf T}: \rho \in \mathcal{D}(\mathcal{H})\},
    \end{equation}
    which is a compact convex subset of $\mathbb{R}^{\kappa}$. 
    For example, the JNR of the Pauli operators $X$, $Y$, $Z$ forms the well-known Bloch ball, with all single qubit states mapping inside. 

We now show that the JNR of set $\{S^{(k)}\}_{k\in[n]}$ is a polytope with at most $n + 1$ vertices. 
Since these operators are diagonal in their common eigenbasis $|\boldsymbol{\lambda}\rangle$, for any state $\rho$ and each $k\in[n]$, we have
\begin{equation}
    \label{eq:4}
    \operatorname{Tr}(\rho S^{(k)})=\sum_{\boldsymbol{\lambda} \in \mathbb{Z}_2^n}\langle\boldsymbol{\lambda}|\rho S^{(k)}|\boldsymbol{\lambda}\rangle=\sum_{\boldsymbol{\lambda}\in \mathbb{Z}_2^n}\underbrace{\langle\boldsymbol{\lambda}|\rho|\boldsymbol{\lambda}\rangle}_{p_{\boldsymbol{\lambda}}}\langle\boldsymbol{\lambda}|S^{(k)}|\boldsymbol{\lambda}\rangle,
\end{equation}
where $p_{\boldsymbol{\lambda}}\geq0$ and $\sum_{\boldsymbol{\lambda}}p_{\boldsymbol{\lambda}}=1$. 
Therefore, the JNR is a polytope with at most $2^n$ vertices. To identify these vertices, we regard $\boldsymbol{\lambda}$ as a spin configuration with $\lambda_i=0$ ($1$) corresponding to spin-up (spin-down). Let $n_{-}\coloneqq\#\{i:\lambda_i=1\}$ denotes the number of spin-down components. Since each eigenvalue $\langle\boldsymbol{\lambda}|S^{(k)}|\boldsymbol{\lambda}\rangle$ is invariant under spin permutations, it depends only on $n_-$. The JNR is then the convex hull of $n+1$ distinct vertices $V_{n_-}\coloneqq (\langle\boldsymbol{\lambda}|S^{(1)}|\boldsymbol{\lambda}\rangle,\ldots,\langle\boldsymbol{\lambda}|S^{(n)}|\boldsymbol{\lambda}\rangle)^{\mathsf T}$, as illustrated in Fig.~\ref{fig:framework}~(b). 
% Hence, the $2^n$ common eigenstates give rise to only $n+1$ distinct vertices, as illustrated in Fig.~\ref{fig:framework}~(b). 
% Since the eigenvalue $\langle\boldsymbol{\lambda}|S^{(k)}|\boldsymbol{\lambda}\rangle$ is invariant under the spin permutations, it depends only on the number of spin-down components $n_{-}\coloneqq\#\{i:\lambda_i=1\}$. 
% Hence, all eigenstates with the same $n_{-}$ map to a single vertex, denoted by $V_{n_-}$. % by $V_{n_-}\coloneqq\bigl(\langle\boldsymbol{\lambda}|S^{(1)}|\boldsymbol{\lambda}\rangle,\ldots,\langle\boldsymbol{\lambda}|S^{(n)}|\boldsymbol{\lambda}\rangle\bigr)^{\mathsf T}$. 
% Using the spin picture, we can easily derive an analytic expression for each component of $V_{n_{-}}$. 
For $k=1$, the eigenvalue %$\langle\boldsymbol{\lambda}|S^{(1)}|\boldsymbol{\lambda}\rangle$ 
corresponds to the total spin, i.e., $V_{n_-}(1)=({n-2n_-})/{n}$. %For $k = n$, it is $V_{n_-}(n)=(-1)^{n_-}$. 
More generally, we have 
\begin{equation}
    \begin{aligned}
    V_{n_{-}}(k) = {1\over L_k} \sum_{|I| = k} (-1)^{\sum_{i\in I}\lambda_i} = {1\over L_k}\sum_{j=0}^k (-1)^j \binom{n_{-}}{j}\binom{n_{+}}{k-j},
    \end{aligned}
\end{equation}
where $n_{+} = n-n_{-}$ and terms with $j>n_-$ vanish. It coincides with the $k$-body correlators of spin states~\cite{prl_2023_jiajie}. Here, we choose $k$ spins from the total of $n$. If $j$ of them are spin-down and the remaining $(k - j)$ are spin-up, the product of their spin values contributes a factor of $(-1)^j$. There are $\binom{n_{-}}{j}\binom{n_{+}}{k-j}$ such combinations. 
% Based on Eq.~\eqref{eq:4}, any point in the JNR admits a convex decomposition of these vertices, i.e., ${\rm tr}\left(\rho S^{(k)}\right) = \sum_{n_- = 0}^n p_{n_-}V_{n_-}(k)$ with $p_{n_{-}}\geq 0$ and $\sum_{n_{-} = 0}^n p_{n_{-}} = 1$. 

\textit{Scalable fidelity bounds from incomplete data---}The simple geometry enables fidelity to be bounded efficiently using the expectation values of only a subset $J\subseteq\{S^{(k)}\}_{k\in[n]}$. 
When these values are estimated from a finite number $N_{\rm tot}$ of single-shot measurements, we replace the exact equality constraints in Eq.~\eqref{eq:boundingfidelity} with inequalities that account for statistical fluctuations~\cite{quantum_2024_acin,prl_2026_acin}. 
Specifically, suppose the measurement statistics yield empirical estimates $\tilde{\mu}_k$ with additive errors at most $\epsilon$ for all $S^{(k)}\in J$ at a confidence level of $1-\delta$. 
We can bound the true fidelity with the same confidence using the following LP with only $n + 1$ variables %The following LP bounds the true fidelity with the same confidence, involving only $n+1$ variables
\begin{equation}
    \begin{aligned}
    \label{eq:main}
       &F_{\min/\max}(J) =  \underset{\{p_{n_-}\}_{n_-=0}^{n}}{\min/\max} ~~~p_0\\
        &~~~~~~{\rm s.t.}~~~ \left|\sum_{n_{-}=0}^n p_{n_{-}}V_{n_{-}}(k) - \tilde{\mu}_k \right| \leq \epsilon,~~~~\forall S^{(k)}\in J,\\
        &~~~~~~~~~~~~~~~p_{n_{-}}\geq 0,~~\sum_{n_{-} = 0}^n p_{n_{-}} = 1,~~\forall n_{-}.
    \end{aligned}
\end{equation}
Here, the objective is the fidelity with the target state $|\psi\rangle=|\boldsymbol{\lambda}=\boldsymbol{0}\rangle$. The constraints use Eq.~\eqref{eq:4} to ensure agreement with the measured observed data within $\epsilon$. 
% Here, the objective function is $p_0$, as the target state is the unique common $+1$ eigenstate, i.e., $|\psi\rangle = |\boldsymbol{\lambda} = \boldsymbol{0}\rangle$. 

% \yl{If measured statistics ensure that the true state satisfies all constraints simultaneously with probability at least $1-\delta$, then its fidelity lies in $[F_{\min}(J),F_{\max}(J)]$ with the same confidence.} 

From the viewpoint of symmetry reduction of an SDP~\cite{rmp_2024_armin}, we reduce the exponential SDP Eq.~\eqref{eq:boundingfidelity} to the linear-size LP Eq.~\eqref{eq:main} by grouping the stabilizing operators. This choice makes both the constraints and the objective invariant under $\rho\mapsto U_\pi\rho U_\pi^\dagger$, where $U_\pi|\boldsymbol{\lambda}\rangle = |\pi\boldsymbol{\lambda}\rangle$ and $\pi\in S_n$ permutes the eigenbasis labels rather than physical qubits. Symmetrizing any feasible state as $\rho_{\rm sym}={1\over n!}\sum_{\pi\in S_n}U_\pi\rho U_\pi^\dagger$ therefore reduces the optimization to the $n+1$ variables.

\textit{Extensions---}Our approach can be extended to certify magic using fidelity-based witnesses~\cite{prxquantum_2026_liu_triangle,arxiv_2026_tang_witness_expansion,npjqi_2026_haug_mixed_magic} and to bound nonlinear quantities. 
In fact, the probabilities $p_{n_-}$ correspond to projections onto $n+1$ mutually orthogonal subspaces of respective dimensions $L_{n_-}$, which constrain the spectrum of $\rho$. 
Thus, replacing the objective function in Eq.~\eqref{eq:main} with $\sum_{n_-}p_{n_-}^2/L_{n_-}$ yields a lower bound on the purity and hence an upper bound on the second-order R\'enyi entropy~\cite{qip_2025_generalizedentropy}.
It turns the LP into a convex quadratic program without changing the number of variables, which can still be solved efficiently.

In addition, our framework naturally extends to different groupings of generators. 
% for some target states with spatially local generators, their generators can be partitioned into $\chi$ groups $A_1, A_2,\cdots ,A_\chi$ such that all products within each group $A_c$ are simultaneously accessible from parallel measurements~\cite{prapplied_2019_hypergraph}. 
For some target states with spatially local generators, the generators can be partitioned into $\chi$ groups such that all products within each group are efficiently accessible~\cite{prapplied_2019_hypergraph}. 
Defining stabilizer moments within each group yields a similar polytope description, with vertices labeled by the number of down spins in each group. We show in Appendix~\ref{app:binary-group-fidelity} that the resulting LP for fidelity bounds can be reduced to $2^\chi$ variables. 
% We can similarly define the grouped stabilizer moments $S_c^{(k)}\coloneqq\frac{1}{\binom{n_c}{k}}\sum_{\substack{I\subseteq A_c\\|I|=k}}\prod_{i\in I}g_i$ for $c\in[\chi]$ and $k\in[n_c]$. 
% With the same grouping of spins, the vertices of the associated JNR depend only on the number of down spins in each group, specified by $\boldsymbol{n}_- = (n_{-,1},\cdots,n_{-,\chi})^{\mathsf T}$. The fidelity can therefore be bounded by solving an LP with $\prod_{c=1}^{\chi}(n_c+1)$ variables. In Appendix~\ref{app:binary-group-fidelity}, we further reduce the number of variables to $2^\chi$.

\textit{Applications for hypergraph states---}A hypergraph state associated with a hypergraph $H=(V,E)$ is defined by $|H\rangle=\left(\prod_{e\in E}C_e\right)|+\rangle^{\otimes n}$, where $C_e\coloneqq\mathbb{I}-2|1\cdots1\rangle_e\langle1\cdots1|$ is the generalized controlled-$Z$ gate. 
It applies a phase flip when all qubits in $e$ are in $|1\rangle$. 
% Here, $V$ labels the qubits and $E$ denotes the hyperedge set. 
Then, its generator is given by~\cite{JPA_2014_Otfried}
\begin{equation}
    \label{eq:generator}
    g_i=\left(\prod_{e\in E}C_e\right)X_i\left(\prod_{e\in E}C_e\right)^\dagger
    =X_i\otimes \prod_{e\in \{e\in E: i \in e\}  }C_{e\setminus\{i\}}. 
\end{equation}
A graph state is the special case in which every hyperedge contains two vertices and is therefore a stabilizer state.

The Pauli expansion of these stabilizing operators may contain exponentially many terms, which results in the difficulty of DFE for hypergraph states~\cite{prapplied_2019_hypergraph}. 
In fact, many of them can be estimated simultaneously from a single local Pauli measurement setting. 
To show this, we rewrite the generalized controlled-Z gate as $C_e=\mathbb{I}-2^{1-|e|}\otimes_{j\in e}(\mathbb{I}_j-Z_j)$, which is a linear combination of Pauli-$Z$ strings on the qubits in $e$. Thus, the generator Eq.~\eqref{eq:generator} can be estimated simultaneously using a single parallel setting such as $X_1\otimes Z_2\otimes \cdots \otimes Z_n$. 
It follows that $g_1g_2$ can be measured from only two settings 
\begin{equation}
\label{eq:common-pauli-settings}
\langle g_1 g_2 \rangle_\sigma \ \Leftarrow\ 
\left\{
\begin{array}{l}
X_1\otimes X_2\\
Y_1\otimes Y_2
\end{array}
\right\}\otimes Z_3\otimes \cdots \otimes Z_n,
\end{equation}
where terms containing an odd number of Pauli $Y=iXZ$ operators vanish due to Hermiticity. 
More generally, $S^{(k)}$ contains $L_k$ stabilizing operators, each of which can be estimated using at most $2^{k-1}$ such local Pauli measurement settings. 
The specific measurements of a given hypergraph structure are detailed in Appendix~\ref{app:large-deviation bounds}. 
% In addition, such a grouping strategy has been shown to improve the estimation accuracy~\cite{quantum_2025_groupingDFE}. 

Since the stabilizer moments $S^{(k)}$ with smaller $k$ have lower measurement overhead, we propose to bound the fidelity by measuring $J = \{S^{(1)}, \cdots, S^{(K)}\}$. 
Here, $K$ is a chosen integer, which can be increased to tighten the fidelity bounds. 
% For each $k$, we estimate $S^{(k)}$ by uniform sampling rather than enumerating all $L_k$ operators, making the sample complexity independent of $n$ at fixed $k$. For each $k=1,\ldots,K$, we independently and uniformly sample $l$ stabilizing operators $s_{\boldsymbol{x}}$ satisfying $|\boldsymbol{x}|=k$. 
For each $k\leq K$, we estimate $S^{(k)}$ by independently and uniformly sampling $l$ of its constituent stabilizing operators. 
This avoids enumerating all $L_k$ terms and makes the sample complexity independent of $n$ at fixed $k$. 
In the worst case, each of these operators requires $2^{k-1}$ settings. %every stabilizing operator contributing to $S^{(k)}$
We then use $M$ copies of states for each setting, and postprocess the outcomes to construct the unbiased estimator $\tilde{\mu}_k$. % of $\langle S^{(k)}\rangle_{\sigma}$
In Appendix~\ref{app:large-deviation bounds} we show using Hoeffding's inequality that
choosing $l=\left\lceil\frac{8\ln(4K/\delta)}{\epsilon^2}\right\rceil$ and $M=\left\lceil\frac{8(1+\epsilon/3)\ln(4K/\delta)}{l\epsilon^2}\right\rceil$ guarantees
% \begin{equation}
%     \Pr\left[\bigcup_{k = 1}^K\left|\tilde{\mu}_k-\langle S^{(k)}\rangle_\sigma\right|\geq\epsilon\right]\leq\delta.
% \label{eq:moment-sampling-guarantee}
% \end{equation}
\begin{equation}
    \Pr\left[\bigcap_{k = 1}^K\left|\tilde{\mu}_k-\langle S^{(k)}\rangle_\sigma\right|\leq\epsilon\right]\geq 1 - \delta.
\label{eq:moment-sampling-guarantee}
\end{equation}
The corresponding total number of single-shot measurements satisfies $\mathbb{E}\left(N_{\rm tot}\right)=O\!\left(2^K\ln(K/\delta)/\epsilon^2\right)$, independent of $n$. 
This sample complexity can be further reduced by exploiting the structure of target states. 
After solving Eq.~\eqref{eq:main}, we obtain the fidelity interval $[F_{\min}(J),F_{\max}(J)]$. Let $\bar{F}$ and $\upDelta$ denote its midpoint and half-width respectively, this scheme also guarantees that the true fidelity lies in $[\bar{F}-\upDelta,\bar{F}+\upDelta]$ with probability at least $1-\delta$.

\textit{Numerical examples---}Intuitively, a small $K$ may suffice to obtain informative fidelity bounds at low noise levels, as suggested by the state verification and certification protocols~\cite{prapplied_2019_hypergraph,np_2025_hyh}. 
% For example, the exact value $\langle S^{(1)}\rangle_\sigma=1$ alone implies unit fidelity. 
Under stronger noise, however, the compatible state space may include states with zero fidelity, yielding trivial bounds (Appendix~\ref{app:binary-group-fidelity}). 
Therefore, the tightness depends on the empirical estimates $\{\tilde{\mu}_k\}_{k = 1}^K$ and the tolerance $\epsilon$. 
To illustrate a practical recipe, we simulate $\tilde{\mu}_k$ under global and local depolarizing noise. 
In Fig.~\ref{fig:measurement_cost}, we solve Eq.~\eqref{eq:main} for each $\epsilon\in\{10^{-2},10^{-3},10^{-4}\}$ to determine the smallest $K$ and the corresponding measurement cost needed to achieve each target precision $\upDelta$. 
% \yls{In Fig.~\ref{fig:measurement_cost}, For each $\epsilon \in \{10^{-2}, 10^{-3}, 10^{-4}\}$, we determine the $K$ required to achieve a desired precision $\upDelta$. }
% We simulate these estimates under global and local depolarizing noise. 
% For each $\epsilon \in \{10^{-2}, 10^{-3}, 10^{-4}\}$, we determine the $K$ required to achieve a desired precision $\upDelta$. 
% Suppose an experiment targets the pure state $|\psi\rangle$ and prepares a state $\sigma$ affected by either noise process,
Here, we consider a noisy preparation $\sigma$ of the pure target state $|\psi\rangle$ under either noise process
\begin{equation}
\begin{aligned}
    \sigma_{g}(\eta)&=(1-\eta)|\psi\rangle\langle\psi|+\frac{\eta}{2^n}\mathbb{I},\quad 
    \sigma_{l}(p)&=\mathcal{D}_1^{\otimes n}\!\left(|\psi\rangle\langle\psi|\right),
\end{aligned}
\label{eq:noise-models}
\end{equation}
where $\mathcal{D}_1(\rho)=(1-p)\rho+\frac{p}{3}\sum_{P\in\{X,Y,Z\}}P\rho P$ is a single qubit depolarizing noise channel.

Under global depolarizing noise, the empirical estimates are independent of the target state and satisfy $\tilde{\mu}_k=1-\eta$ for all $k\geq1$. 
In the inset of Fig.~\ref{fig:measurement_cost}~(a), we observe that a small $K$ suffices to restrict the fidelity to a small interval even when the true fidelity is as low as $0.5$. 
In the main panel, we choose $\eta$ to fix the true fidelity at $0.6$. 
For arbitrary hypergraph target states, the upper bound on the expected measurement cost $\mathbb{E}(N_{\rm tot})$ increases only slowly with $n$ for all precisions $\upDelta$ considered. 
For comparison, we consider a worst-case example for standard DFE, which is given by a hypergraph state with a single hyperedge containing all $n$ qubits. To achieve the same fidelity interval $\upDelta$ and confidence level $1-\delta$, standard DFE requires $\mathbb{E}(N_{\rm DFE}) \approx 1 + {8\over \upDelta^2 \delta} + {4(2^n - 1)\over \upDelta^2}\ln({4\over \delta})$ single-shot measurements~\cite{prapplied_2019_hypergraph}.

\begin{figure}[t]
\vspace{-0.2cm}
    \centering
\includegraphics[width=0.9\linewidth]{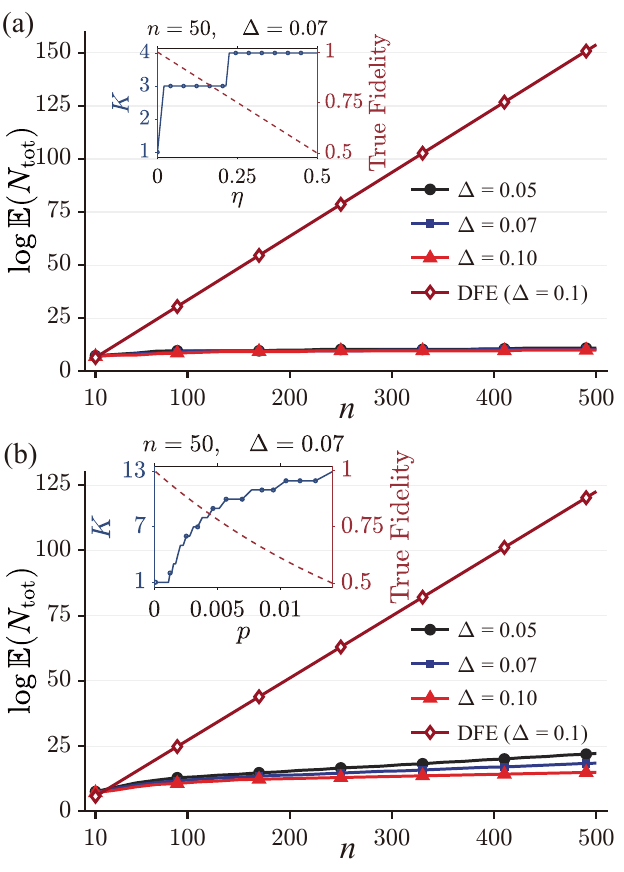}
    \caption{\label{fig:measurement_cost}
    Sample complexity for fidelity bounding under (a) global depolarizing noise for arbitrary hypergraph target states and (b) local depolarizing noise for the 1D order-3 cluster state. Based on statistics $\tilde{\mu}_k$ simulated under the noise model, the main plots show the sample complexity $\mathbb{E}(N_{\rm tot})$ required to obtain a fidelity interval of $2\upDelta$. In these plots, the noise levels $\eta$ and $p$ are chosen to give a true fidelity of $F=0.6$, and the failure probability is fixed at $\delta=0.01$. 
    The insets show the number of constraints in Eq.~\eqref{eq:main} required to achieve a precision of $\upDelta=0.07$ at different noise levels, with $n=50$. They also show the corresponding true fidelity. 
    All logarithms in this work are base $10$. 
    See Appendix~\ref{app:fig2-procedure} for the numerical procedure.}
    \vspace{-0.5cm}
\end{figure}

Under local depolarizing noise, the empirical estimates depend on $|\psi\rangle$. 
In Fig.~\ref{fig:measurement_cost}~(b), we consider a one-dimensional (1D) order-3 cluster state $|C\rangle \coloneqq \left(\prod_{i=1}^{n-2}C_{\{i,i+1,i+2\}}\right)|+\rangle^{\otimes n}$. 
We simulate the statistics $\tilde{\mu}_k$ and evaluate the sample complexity of standard DFE for this state, see Appendices~\ref{sec:supp-local-depolarizing} for details.
The inset shows that stronger noise requires larger $K$ to achieve the same precision. 
The main panel shows that our expected measurement cost is significantly below that of the standard DFE. 
In our scheme, the sample complexity satisfying Eq.~\eqref{eq:moment-sampling-guarantee} is reduced to $\mathbb{E}\left(N_{\rm tot}\right)\leq Ml\min\left\{2^K-1,\frac{n}{2K}\left[(1+\frac{2K}{n})^K-1\right]\right\}$ for this specific state (Appendix~\ref{app:collision-adaptive-moments}). 

This advantage partly stems from estimating many compatible Pauli observables by parallel measurements, a benefit also demonstrated in grouping-based DFE~\cite{quantum_2025_groupingDFE}. 
A further saving is that, we only need to estimate a subset of stabilizer moments to obtain sufficiently tight fidelity bounds. 
This is because the positivity imposes a powerful constraint~\cite{npjqi_2015_kalev_positivity} that is reflected in the JNR geometry, allowing the measured stabilizer moments to bound the unmeasured ones.

\begin{figure}[t]
    \centering
\includegraphics[width=1\linewidth]{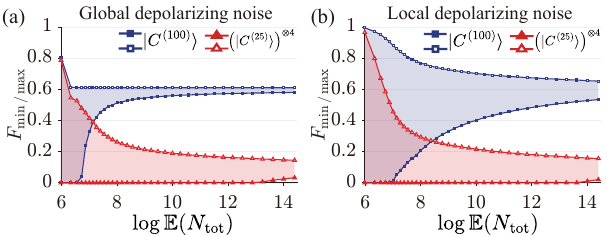}
\caption{\label{fig:topology_certification}Distinguishing different hypergraph topologies using fidelity intervals obtained from the same local Pauli measurement data. 
Suppose the prepared state $\sigma$ is obtained from the connected $100$-qubit 1D order-3 cluster state $|C^{(100)}\rangle$ under (a) global or (b) local depolarizing noise. The noise strength is chosen such that the true fidelity is $\langle C^{(100)}|\sigma|C^{(100)}\rangle=0.6$.
% A noisy preparation $\sigma$ of the connected 100-qubit 1D order-3 cluster state $|C^{(100)}\rangle$ is evaluated with respect to 
By choosing $\epsilon=\delta=0.01$ and $K=1,\ldots,35$, we show the fidelity interval for two different targets: the connected state itself (blue squares) and the disconnected candidate $\bigl(|C^{(25)}\rangle\bigr)^{\otimes4}$ (red triangles). 
% Under both (a) global and (b) local depolarizing noise, the noise strength is chosen such that the true fidelity with the connected target is $F=\langle C^{(100)}|\sigma|C^{(100)}\rangle=0.6$. 
Filled and open markers denote the lower and upper fidelity bounds, respectively, and the shaded regions show the corresponding intervals. 
} 
    \vspace{-0.5cm}
\end{figure}

Finally, since the parallel measurement settings for estimating $S^{(k)}$ are independent of the hypergraph structure, we can simultaneously bound the fidelities of multiple target states from the same measurement like Eq.~\eqref{eq:common-pauli-settings}. 
We illustrate this with two target states: the connected $100$-qubit 1D order-$3$ cluster state $|C^{(100)}\rangle$, and the product of four disconnected $25$-qubit 1D order-$3$ cluster state $\bigl(|C^{(25)}\rangle\bigr)^{\otimes4}$. 
In Fig.~\ref{fig:topology_certification}, we simulate the measurement statistics $\tilde{\mu}_k$ of a noisy preparation of $|C^{(100)}\rangle$ under (a) global and (b) local depolarizing noise (Appendix~\ref{sec:supp-local-depolarizing}). 
% Including higher-order stabilizer moments gives tighter fidelity bounds $F_{\min / \max}$ for both targets. 
We find that even when the fidelity intervals remain broad, the bounds already indicate a preference for one candidate target over the other.

\textit{Analytic fidelity bounds---}The polytope formulation also provides a unified framework for deriving analytic fidelity bounds from limited stabilizing data. 
This is because the equality-constrained version of Eq.~\eqref{eq:main} has an optimal solution with at most $|J|+1$ nonzero probabilities. 
For the choices $J=\{S^{(1)}\}$ and $J=\{S^{(1)},S^{(n)}\}$, the analytical solutions generalize the known fidelity bounds for stabilizer states~\cite{pra_2006_boundingfidelity,pra_2019_Kalev_BoundingFidelity,jmo_2009_boundingfidelity,pra_2015_boundingfidelity} to arbitrary pure states. 
We further provide new bounds for $J=\{S^{(1)},S^{(n-1)},S^{(n)}\}$,
\begin{equation}
\label{eq:analytic-three-moment-bound}
F_{\min}(J)\geq {1\over 8}
\left\{(4-n)\!\left[\langle S^{(n)}\rangle_\sigma+1\right]
+n\!\left[\langle S^{(1)}\rangle_\sigma+\langle S^{(n-1)}\rangle_\sigma\right] \right\}. 
\end{equation}
For stabilizer states, this bound requires only $2n+1$ Pauli measurement settings. See Appendix~\ref{app:dual} for details.

\textit{Conclusion---}We have developed a framework for bounding the fidelity to non-stabilizer states from incomplete data. 
We group the stabilizing operators into the stabilizer moments $S^{(k)}$ to introduce a permutation symmetry of the common eigenbasis labels. This symmetry yields simplifications typically associated with permutationally invariant systems, without requiring any assumptions about either the target or the experimental state. 
With this we obtain tight fidelity bounds for hypergraph states with hundreds of qubits. We demonstrate a substantial reduction in measurement cost compared with standard DFE through numerical examples under global and local depolarizing noise. 
% In addition, the same measurement record also allows candidate hypergraph topologies to be distinguished through separated fidelity intervals. 
In addition, our framework allows different hypergraph topologies to be distinguished using shared measurement data. 
Beyond numerical bounds, the polytope formulation unifies analytic fidelity inequalities for specific state families and extends them to arbitrary pure targets, enabling the systematic derivation of new inequalities and entanglement witnesses. 
% Beyond numerical bounds, the simple polytope formulation unifies analytic fidelity inequalities previously derived for specific state families and extends them to arbitrary pure target states, enabling the systematic derivation of new bounds and fidelity-based entanglement witnesses. 

Our work further shows that the physical constraint of positivity provides a powerful tool for certification from incomplete data. 
Future work could explore alternative groupings of operators and the symmetries they induce to address a broader range of optimization tasks, including those arising in non-stabilizer resource characterization and quantum cryptography. 
Furthermore, replacing the linear fidelity objective with nonlinear functions may extend our framework to bounding nonlinear magic witnesses. 
Our results suggest a broader role for positivity and symmetry in the design of efficient quantum certification protocols. 

% \yls{Positivity imposes powerful constraints that allow measured stabilizer moments to bound unmeasured ones.} 
% The associated measurement overhead depends on the Pauli expansion of the chosen stabilizing operators. 
% Future work could combine Pauli propagation with circuit symmetries to reduce this overhead~\cite{arxiv_2025_teng_symmetry_merging}. 
% Under local noise models, the contribution of high-weight Pauli terms may allow the Pauli expansion to be truncated with a rigorous error bound~\cite{prxquantum_2026_angrisani_local_noise}. 
% Furthermore, exploring alternative groupings of stabilizing operators and the symmetries they induce may lead to scalable methods for characterizing non-stabilizer resources. 
% Combining the certified purity bounds with additional measurements could also support nonlinear magic resource witnesses~\cite{npjqi_2026_haug_mixed_magic,arxiv_2026_tang_witness_expansion}.

% \begin{acknowledgments}
\textit{Acknowledgments---}This work is supported by National Natural Science Foundation of China (No.~12125402, No.~12534016, No.~12605001, No.~123B2073), Beijing Natural Science Foundation (Grant No.~Z240007), Quantum Science and Technology-National Science and Technology Major Project (Grant No.~2024ZD0302401 and No.~2021ZD0301500), 
the Deutsche Forschungsgemeinschaft (DFG, German Research Foundation, project number 563437167), the Sino-German Center for Research Promotion (Project M-0294), and the German Federal Ministry of Research, Technology and Space (Project BeRyQC, Grant No.~13N17292). 
Y. L. acknowledges the China National Postdoctoral Program for Innovative Talents (No.~BX2026036), and the China Scholarship Council (Grant No.~202306010353). 
% \end{acknowledgments}

\bibliographystyle{apsrev4-2}
% REVTeX normally writes the style only at \end{document}.  Write it here so
% BibTeX can run immediately, while retaining REVTeX's APS control records.
\begingroup
\makeatletter
\@booleanfalse\bibliographystyle@sw
\write@bibliographystyle
\endgroup
\bibliography{Refs}

\clearpage
\onecolumngrid

\appendix

\setcounter{secnumdepth}{2}

% Show only appendix sections and subsections in this table of contents.
\begingroup
\makeatletter
\renewcommand{\tocname}{Contents}
\let\appendixcontentsline\contentsline
\renewcommand{\contentsline}[4]{}
\renewcommand{\appendix@toc}{\let\contentsline\appendixcontentsline}
\let\l@subsubsection\@gobbletwo
\tableofcontents
\endgroup

\section{Local Pauli measurement overhead}\label{app:one-step-tfim}

\subsection{Pauli propagation for circuit families}
\label{app:pauli-propagation-circuit-families}

We first consider several circuit families for which the Pauli support of each order-$k$ stabilizing operator admits a bound independent of $n$. For a target state $|\psi\rangle=U|0\rangle^{\otimes n}$ with generators $g_i=UZ_iU^\dagger$, we write
\begin{equation}
g_I\coloneqq\prod_{i\in I}g_i=UZ_IU^\dagger,
\qquad Z_I\coloneqq\prod_{i\in I}Z_i,
\qquad |I|=k.
\label{eq:general-circuit-stabilizer}
\end{equation}
Let $N_{\mathrm P}(A)$ denote the number of nonzero coefficients in the expansion of $A$ in the Pauli basis $\{\mathbb{I},X,Y,Z\}^{\otimes n}$. Suppose $A$ acts nontrivially on at most $m$ qubits, then we directly give an upper bound of $N_{\mathrm P}(A)\leq4^m$.

\textbf{Clifford-enhanced product states.} The Clifford-enhanced product state ($\mathcal{C}$PS) family~\cite{arxiv_2026_sater_cps_certification,PRXQuantum_2025_cps} consists of states of the form
\begin{equation}
|\psi\rangle=C\left(\bigotimes_{j=1}^n u_j|0\rangle\right),
\qquad U=C\left(\bigotimes_{j=1}^n u_j\right),
\label{eq:cps-target-circuit}
\end{equation}
where the $u_j$ are arbitrary single-qubit unitaries and $C$ is an arbitrary Clifford circuit. Such states are obtained by applying a Clifford circuit to a product of arbitrary single-qubit states, which includes magic states. 
Since $u_jZ_ju_j^\dagger=a_jX_j+b_jY_j+c_jZ_j$, with real coefficients satisfying $a_j^2+b_j^2+c_j^2=1$, we obtain
\begin{equation}
g_I=C\left[\prod_{j\in I}(a_jX_j+b_jY_j+c_jZ_j)\right]C^\dagger.
\label{eq:cps-stabilizer-expansion}
\end{equation}
The product contains at most $3^k$ Pauli strings. Clifford gates preserve this number. Consequently, we have
\begin{equation}
N_{\mathrm P}(g_I)\leq3^k,
\label{eq:pauli-support-cps}
\end{equation}
which is independent of $n$ and the depth of $C$.

\textbf{Arbitrary two-qubit circuits with circuit depth $d_{\mathrm c}$.} Suppose $U$ has $d_{\mathrm c}$ layers of disjoint one- or two-qubit gates. Definition~$16$ and Fact~$5$ of Ref.~\cite{stoc_2024_huang_shallow_circuits} give a backward-lightcone size of at most $2^{d_{\mathrm c}}$ for a single qubit. Applying this statement to $U^\dagger$ bounds the support of $g_i=UZ_iU^\dagger$. Equivalently, each layer can at most double the number of qubits on which the operator $g_i$ acts nontrivially. Taking the union of the supports of the $k$ generators therefore yields
\begin{equation}
N_{\mathrm P}(g_I)\leq4^{\min\{n,\,k2^{d_{\mathrm c}}\}}.
\label{eq:pauli-support-arbitrary-circuit}
\end{equation}
For shallow circuits, the depth is $d_{\mathrm c}=O(1)$ and we have $N_{\mathrm P}(g_I)\leq 2^{O(k)}$.

\textbf{Arbitrary two-qubit nearest-neighbor circuits with circuit depth $d_{\mathrm c}$.} Consider a depth-$d_{\mathrm c}$ circuit on a chain in which two-qubit gates act only on neighboring qubits. The support of $g_i$ is therefore confined to the sites $\{i-d_{\mathrm c},\ldots,i,\ldots,i+d_{\mathrm c}\}$ that lie on the chain, acting on at most $2d_{\mathrm c}+1$ qubits. Taking the union of these supports for the $k$ generators in $g_I$ gives at most $\min\{n,(2d_{\mathrm c}+1)k\}$ qubits, yielding
\begin{equation}
N_{\mathrm P}(g_I)\leq4^{\min\{n,\,k(2d_{\mathrm c}+1)\}}.
\label{eq:pauli-support-one-dimensional}
\end{equation}
Thus, the exponent in this Pauli-count bound grows linearly with depth. 

% \yl{\textbf{Paired Bardeen--Cooper--Schrieffer (BCS) states.} Consider an unprojected BCS pair product~\cite{qst_2021_khamoshi_agp} with real amplitudes in the occupation-number encoding of $n/2$ disjoint mode pairs $(a_j,b_j)$,
% \begin{equation}
% |\psi\rangle=\bigotimes_{j=1}^{n/2}\left(\cos\frac{\theta_j}{2}|00\rangle_{a_jb_j}+\sin\frac{\theta_j}{2}|11\rangle_{a_jb_j}\right).
% \label{eq:bcs-pair-product}
% \end{equation}
% Each pair is prepared by $W_{ab}=e^{-i\theta Y_aX_b/2}$. Conjugating its two Pauli-$Z$ generators gives
% \begin{equation}
% g_a=\cos\theta\,Z_a+\sin\theta\,X_aX_b,\qquad g_b=\cos\theta\,Z_b-\sin\theta\,Y_aY_b,\qquad g_ag_b=Z_aZ_b.
% \label{eq:bcs-pair-generators}
% \end{equation}
% Thus, only pairs with exactly one generator selected in $I$ can contribute two Pauli terms. Denoting the number of such pairs by $s(I)$, we obtain
% \begin{equation}
% N_{\mathrm P}(g_I)\leq2^{s(I)}\leq2^{\min\{k,n-k\}},
% \label{eq:pauli-support-bcs}
% \end{equation}
% since each such pair contains one selected and one unselected qubit. The same Pauli-count bound holds after applying an arbitrary Clifford circuit to the target state.}

These results give general upper bounds on the number of local Pauli settings. They do not include further reductions from specific gate structures, cancellations during Pauli propagation, or parallel measurements of compatible terms. 
Below, we provide tighter bounds for hypergraph and weighted graph states, as well as Trotterized Ising states. 
In fact, circuit symmetries can also reduce such a number through orbit merging~\cite{arxiv_2025_teng_symmetry_merging}, and high-weight Pauli terms can be suppressed by noise~\cite{prxquantum_2026_angrisani_local_noise}.

\subsection{Parallel measurements for the stabilizing operators of hypergraph states}
\label{app:hypergraph-pauli-settings}

For a hypergraph state $|H\rangle=U_H|+\rangle^{\otimes n}$, where $U_H=\prod_{e\in E}C_e$, each generator takes the form
\begin{equation}
g_i=U_HX_iU_H^\dagger
=X_i\otimes\prod_{\substack{e\in E\\i\in e}}C_{e\setminus\{i\}} = X_i\otimes\left( \text{a linear combination of Pauli-$Z$ strings}\right).
\label{eq:hypergraph-generator-diagonal}
\end{equation}
All these terms can be estimated from a single parallel setting, with $X$ measured on qubit $i$ and $Z$ on all other qubits. 

It follows that $g_1g_2$ can be measured from only two parallel settings:
\begin{equation}
\label{eq:hypergraph-two-generator-settings}
\langle g_1 g_2 \rangle_\sigma \ \Leftarrow\
\left\{
\begin{array}{l}
X_1\otimes X_2\\
Y_1\otimes Y_2
\end{array}
\right\}\otimes Z_3\otimes \cdots \otimes Z_n,
\end{equation}
where terms containing an odd number of Pauli $Y=iXZ$ operators vanish because $g_1g_2$ is real and Hermitian.

More generally, an order-$k$ stabilizing operator is $g_I=\prod_{i\in I}g_i$, where $I\subseteq[n]$ and $|I|=k\geq1$.
Multiplying the expansions of the $g_i$ gives one $X$ factor on each qubit in $I$, together with possible $Z$ factors. Each resulting Pauli term therefore contains $X$ or $Y$ on $I$ and $\mathbb{I}$ or $Z$ elsewhere. Still, since $g_I$ is real and Hermitian, at most $2^{k-1}$ parallel local Pauli settings suffice to estimate each order-$k$ stabilizing operator. This bound is independent of $n$ and the hypergraph structure, and the hypergraph structure enters only the classical processing of the outcomes. In the following, Lemma~\ref{lem:setting_term_parseval} gives the full decomposition.

\subsection{Parallel measurements for weighted graph states}
\label{app:weighted-graph-pauli-settings}

The same support argument extends to weighted graph states~\cite{review_2006_graphstate}. The states have the form $|\psi_D\rangle=D|+\rangle^{\otimes n},\qquad D=\sum_{\boldsymbol z\in\{0,1\}^n}e^{i\phi(\boldsymbol z)}|\boldsymbol z\rangle\langle\boldsymbol z|$, where $\phi$ is real, and the stabilizing operators are $g_I=DX_ID^\dagger$. It only flips the qubits in $I$
\begin{equation}
g_I|\boldsymbol z\rangle=e^{i[\phi(\boldsymbol z\oplus\boldsymbol 1_I)-\phi(\boldsymbol z)]}|\boldsymbol z\oplus\boldsymbol 1_I\rangle,
\end{equation}
where $\boldsymbol 1_I$ is the indicator bit string of $I$. Then, for any local Pauli string $P = P_1\otimes \cdots \otimes P_n$ with $P_{j \in I}\in \{\mathbb{I},Z\}$, we have ${\rm tr}(P g_I) = 0$. 
Thus, only strings with $X$ or $Y$ on $I$ and $\mathbb{I}$ or $Z$ on $\bar I=[n]\setminus I$ contribute. Grouping them by the subset $A\subseteq I$ carrying $Y$ gives
\begin{equation}
g_I=\sum_{A\subseteq I}\underbrace{Y\otimes\cdots\otimes Y}_{\text{qubits in }A}\otimes\underbrace{X\otimes\cdots\otimes X}_{\text{qubits in }I\setminus A}\otimes \underbrace{\left(\text{a linear combination of Pauli-$Z$ strings}\right)}_{\text{qubits in }[n]\setminus I},
\label{eq:diagonal-stabilizer-decomposition}
\end{equation}
where the coefficients are real because $g_I$ is Hermitian. Thus, at most $2^k$ parallel local Pauli settings suffice, independently of circuit depth or interaction geometry, even when the Pauli support is large.

\subsection{Pauli support of the stabilizing operators of trotterized transverse-field Ising states}
\label{app:tfim-pauli-overhead}

We now specialize to target states generated by Trotterized transverse-field Ising dynamics, which has been implemented as a many-body benchmark~\cite{nature_2023_kim,prl_2026_learningmixedstates}.
In particular, for the time dynamics of the Hamiltonian
\begin{equation}
H_{\mathrm{TFIM}}=H_{ZZ}+H_X,
\qquad
H_{ZZ}=-J\sum_{j=1}^{n-1}Z_jZ_{j+1},
\qquad 
H_X=h\sum_{j=1}^{n}X_j,
\end{equation}
where $J>0$ is  the coupling of nearest-neighbour spins and $h$ is  global transverse field. 
Let $\delta t$ denote the time step in the Trotter decomposition. 
For $J\delta t=\pi/4$ and $h\delta t=\pi/8$, spin dynamics from an initial state can be simulated by means of the Trotter decomposition of the time-evolution operator,
\begin{equation}
\begin{aligned}
E \coloneqq \prod_{j=1}^{n-1}e^{i\pi Z_jZ_{j+1}/4}, \qquad 
R \coloneqq \prod_{j=1}^{n}e^{-i\pi X_j/8}.
\end{aligned}
\end{equation}
The state can be simulated using gates $R_X(\theta)=e^{-i\theta X/2}$ and $R_{ZZ}(\theta)=e^{-i\theta Z\otimes Z/2}$ with angles of $\theta_h=\pi/4$ and $\theta_J=-\pi/2$, respectively. 
For $D \in \mathbb{Z}$, let the target state be
\begin{equation}
|\psi_D\rangle\coloneqq (ER)^D|0\rangle^{\otimes n} = \left(\prod_{j=1}^{n-1}e^{i\pi Z_jZ_{j+1}/4}\prod_{j=1}^{n}e^{-i\pi X_j/8}\right)^D|0\rangle^{\otimes n}.
\end{equation}
The corresponding generalized generators are $g_i^{[D]}\coloneqq(ER)^D Z_i\left[(ER)^D\right]^\dagger$. 

\begin{proposition}
\label{thm:one-step-tfim-pauli}
For $D=1$, consider any subset $I\subseteq[n]$ with $|I|=k\geq1$ and define $g_I^{[1]}\coloneqq\prod_{i\in I}g_i^{[1]}$.
The Pauli expansion of $g_I^{[1]}$ contains at most $2^k$ nonzero terms.
\end{proposition}

\begin{proof}
Writing $Z_I\coloneqq\prod_{i\in I}Z_i$ and using $R_X(\theta)= \cos {(\theta/2)} \mathbb{I} - i\sin (\theta/2)X$, we have 
$RZ_iR^\dagger=(Z_i-Y_i)/\sqrt{2}$, and
\begin{equation}
\begin{aligned}
g_I^{[1]}
&=ERZ_IR^\dagger E^\dagger
 = E \left( \prod_{i \in I} {Z_i - Y_i \over \sqrt{2}} \right)E^\dagger=\frac{1}{2^{k/2}}
\sum_{S\subseteq I}(-1)^{|S|}
E\left(
\prod_{i\in I\setminus S}Z_i
\prod_{i\in S}Y_i
\right)E^\dagger.
\end{aligned}
\label{eq:one-step-tfim-pauli-expansion}
\end{equation}
Since $E$ is a Clifford unitary, the Pauli expansion of $g_I^{[1]}$ contains at most $2^k$ distinct strings.
\end{proof}

\begin{corollary}
\label{cor:two-step-tfim-pauli}
For $D=2$, consider any subset $I\subseteq[n]$ with $|I|=k\geq1$ and define $g_I^{[2]}\coloneqq\prod_{i\in I}g_i^{[2]}$.
The Pauli expansion of $g_I^{[2]}$ contains at most $10^k$ nonzero terms. 
\end{corollary}

\begin{proof}
Since $E$ commutes with $Z_i$, we have $g_i^{[1]}=\frac{1}{\sqrt{2}}\left(Z_i-EY_iE^\dagger\right)$. It follows that
\begin{equation}
g_i^{[2]}=ERg_i^{[1]}R^\dagger E^\dagger=\frac{1}{\sqrt{2}}E\left[RZ_iR^\dagger-R\left(EY_iE^\dagger\right)R^\dagger\right]E^\dagger.
\label{eq:two-step-generator-expansion}
\end{equation}
All factors in $E$ commute, and those not incident on site $i$ commute with $Y_i$ and cancel in the conjugation. Using $e^{i\pi P/4}Qe^{-i\pi P/4}=iPQ$ for anticommuting Pauli operators $P$ and $Q$, we obtain, for $2\leq i\leq n-1$,
\begin{equation}
\begin{aligned}
EY_iE^\dagger
=e^{i\pi Z_{i-1}Z_i/4}(e^{i\pi Z_iZ_{i+1}/4}Y_i e^{-i\pi Z_iZ_{i+1}/4})e^{-i\pi Z_{i-1}Z_i/4}=e^{i\pi Z_{i-1}Z_i/4}\left(X_iZ_{i+1}\right)e^{-i\pi Z_{i-1}Z_i/4} =-Z_{i-1}Y_iZ_{i+1}.
\end{aligned}
\end{equation}
Thus, $EY_iE^\dagger$ is a single Pauli string of Pauli weight at most three.
Using the fact that $R_X(\theta)$ is a single-qubit rotation, we have
\begin{equation}
RX_iR^\dagger=X_i,
\qquad
RY_iR^\dagger=\frac{Y_i+Z_i}{\sqrt{2}},
\qquad
RZ_iR^\dagger=\frac{Z_i-Y_i}{\sqrt{2}}.
\end{equation}
Consequently, $\left[RZ_iR^\dagger-R\left(EY_iE^\dagger\right)R^\dagger\right]$ in Eq.~\eqref{eq:two-step-generator-expansion} has at most $2+2^3=10$ nonzero terms in its Pauli expansion, and $g_i^{[2]}$ also contains at most ten strings. Expanding the product of $k$ such generators yields at most $10^k$ Pauli strings.
\end{proof}

\section{Reduction of the fidelity-bounding SDP to an LP\label{app:A}}
Since for every $s\in\mathcal S$, the action of $\rho \rightarrow s\rho s^\dagger$ preserves both the constraints and the objective in Eq.~\eqref{eq:boundingfidelity} in the main text. 
Hence, every feasible state $\rho$ can be replaced by the twirled state $\bar{\rho}=|\mathcal{S}|^{-1}\sum_{s\in\mathcal{S}}s\rho s^{\dagger}$ without changing its feasibility or objective value.
Since $s_{\boldsymbol{x}}|\boldsymbol{\lambda}\rangle=(-1)^{\boldsymbol{x}\cdot\boldsymbol{\lambda}}|\boldsymbol{\lambda}\rangle$, twirling a matrix element gives
\begin{equation}
    |\mathcal{S}|^{-1}  \sum_{\boldsymbol{x}}s_{\boldsymbol{x}}|\boldsymbol{\lambda}\rangle \langle \boldsymbol{\mu}|s_{\boldsymbol{x}}^\dagger= 2^{-n}\sum_{\boldsymbol{x}\in\mathbb{Z}_2^n}(-1)^{\boldsymbol{x}\cdot(\boldsymbol{\lambda}+\boldsymbol{\mu})}|\boldsymbol{\lambda}\rangle\langle\boldsymbol{\mu}|=\delta_{\boldsymbol{\lambda},\boldsymbol{\mu}}|\boldsymbol{\lambda}\rangle\langle\boldsymbol{\mu}|
\end{equation}
Thus, the twirled state $\bar{\rho}=\sum_{\boldsymbol{\lambda}\in\mathbb{Z}_2^n}p_{\boldsymbol{\lambda}}|\boldsymbol{\lambda}\rangle\langle\boldsymbol{\lambda}|$ is diagonal, where $p_{\boldsymbol{\lambda}}=\langle\boldsymbol{\lambda}|\rho|\boldsymbol{\lambda}\rangle$, $p_{\boldsymbol{\lambda}}\geq 0$, and $\sum_{\boldsymbol{\lambda}}p_{\boldsymbol{\lambda}}=1$.
Since the target state corresponds to $|\boldsymbol{0}\rangle$ in this basis, the SDP reduces exactly to an LP:
\begin{equation}
    \begin{aligned}
    \label{eq:LP1}
    &F_{\min} \left(\mathcal{M}\right) = \min_{\{p_{\boldsymbol{\lambda}}\}\in \mathcal{P}_{2^n} } \left\{p_{\boldsymbol{0}} : \sum_{\boldsymbol{\lambda}\in \mathbb{Z}_2^n} (-1)^{\boldsymbol{x}\cdot\boldsymbol{\lambda}} p_{\boldsymbol{\lambda}} = \langle s_{\boldsymbol{x}}\rangle_\sigma ,~\boldsymbol{x}\in\mathcal{M}\right\},
    \end{aligned}
\end{equation}
where $\mathcal{P}_{2^n} \coloneqq \left\{ \boldsymbol{p}\in\mathbb{R}^{2^n} : p_{\boldsymbol{\lambda}}\ge 0,\ \sum_{\boldsymbol{\lambda}\in\mathbb{Z}_2^n} p_{\boldsymbol{\lambda}} = 1 \right\}$.

\section{Fidelity bounds from grouped stabilizer moments}
\label{app:binary-group-fidelity}

When all generator products within each group can be measured simultaneously, we can bound the fidelity using the grouped stabilizer moments. Here we illustrate this extension for three groups, and further reduce the LP from $(n_1+1)(n_2+1)(n_3+1)$ variables to $2^3$ variables. The intuition is that, since fidelity is the all-up probability, we retain only whether each group is all-up. The full group distributions can still be recovered within these outcomes without changing their probabilities, so the optimal fidelity bounds remain unchanged.

Let $A_1,A_2,A_3$ partition $[n]$, with $n_c=|A_c|>0$, such that all generator products within each group are accessible from the same local Pauli data. As in Eq.~\eqref{eq:fidelity_spinmoment}, define
\begin{equation}
S_c^{(k)}=\frac{1}{\binom{n_c}{k}}\sum_{\substack{I\subseteq A_c\\|I|=k}}\prod_{i\in I}g_i,\qquad c=1,2,3,\quad k=1,\ldots,n_c.
\label{eq:grouped-moment-definition}
\end{equation}
Here $S_c^{(0)}=\mathbb I$. Write $q_{\boldsymbol n_-}=q_{n_{-,1},n_{-,2},n_{-,3}}$ for the joint probability of the three down-spin counts $\boldsymbol n_-=(n_{-,1},n_{-,2},n_{-,3})$, with $0\leq n_{-,c}\leq n_c$. 
In the spin picture, the vertices $V_{\boldsymbol n_-}$ of the associated JNR can be defined in a similar way as in the main text. The fidelity bounds are given by
\begin{equation}
\begin{aligned}
F_{\min/\max}&=\underset{\{q_{\boldsymbol n_-}\geq0,\ \sum_{\boldsymbol n_-}q_{\boldsymbol n_-}=1\}}{\min/\max}\ q_{0,0,0},\\
\mathrm{s.t.}\quad
% &\sum_{\boldsymbol n_-}q_{\boldsymbol n_-}V_{\boldsymbol n_-}=\bigl(\langle S_1^{(1)}\rangle_\sigma,\ldots,\langle S_1^{(n_1)}\rangle_\sigma,\langle S_2^{(1)}\rangle_\sigma,\ldots,\langle S_2^{(n_2)}\rangle_\sigma,\langle S_3^{(1)}\rangle_\sigma,\ldots,\langle S_3^{(n_3)}\rangle_\sigma\bigr)^{\mathsf T}.
&\sum_{\boldsymbol n_-}q_{\boldsymbol n_-}v_c^{(k)}(n_{-,c})=\langle S_c^{(k)}\rangle_\sigma,~~~~~c=1,\ldots,\chi,\quad k=1,\ldots,n_c.
\end{aligned}
\label{eq:grouped-moment-full-lp}
\end{equation}
Here, $v_c^{(k)}(m)=\langle\boldsymbol{\lambda}|S_c^{(k)}|\boldsymbol{\lambda}\rangle$ is the eigenvalue of a common eigenstate $|\boldsymbol{\lambda}\rangle$ with $m$ down spins in group $c$. As shown in the manuscript, it is given by the Kravchuk polynomials $v_c^{(k)}(m)={\mathcal K_k(m;n_c,2)}/{\binom{n_c}{k}}$. 
To interpret this LP, let $X_c\in\{0,\ldots,n_c\}$ be the number of down spins in group $c$ under the candidate joint distribution. Then
\begin{equation}
    q_{\boldsymbol n_-}=\Pr(X_1=n_{-,1},X_2=n_{-,2},X_3=n_{-,3}),\qquad F=q_{0,0,0}=\Pr(X_1=0,X_2=0,X_3=0).
\end{equation}
Each moment constraint in Eq.~\eqref{eq:grouped-moment-full-lp} therefore reads $\sum_{\boldsymbol n_-}q_{\boldsymbol n_-}v_c^{(k)}(n_{-,c})=\sum_{m=0}^{n_c}v_c^{(k)}(m)\underbrace{\textstyle\sum_{\boldsymbol n_-:\,n_{-,c}=m}q_{\boldsymbol n_-}}_{\Pr(X_c=m)}=\langle S_c^{(k)}\rangle_\sigma$. 
The discrete orthogonality of Kravchuk polynomials give the inverse relation 
\begin{equation}
\Pr(X_c=m)=\frac{\binom{n_c}{m}}{2^{n_c}}\sum_{k=0}^{n_c}\binom{n_c}{k}v_c^{(k)}(m)\langle S_c^{(k)}\rangle_\sigma\equiv q_c(m),\qquad m=0,\ldots,n_c.
\end{equation}
Thus, Eq.~\eqref{eq:grouped-moment-full-lp} is equivalent to 
% The measured moments $\langle S_c^{(k)}\rangle_\sigma$, for $k=0,\ldots,n_c$, uniquely fix the marginal distribution of $X_c$, denoted by $q_c(n_{-,c})$. Thus, Eq.~\eqref{eq:grouped-moment-full-lp} can be written equivalently as
\begin{equation}
\begin{aligned}
F_{\min/\max}&=\underset{\{q_{\boldsymbol n_-}\geq0,\ \sum_{\boldsymbol n_-}q_{\boldsymbol n_-}=1\}}{\min/\max}\ \Pr(X_1=0,X_2=0,X_3=0),\\
\mathrm{s.t.}\quad
&\Pr(X_1=n_{-,1})=q_1(n_{-,1}),\qquad n_{-,1}=0,\ldots,n_1,\\
&\Pr(X_2=n_{-,2})=q_2(n_{-,2}),\qquad n_{-,2}=0,\ldots,n_2,\\
&\Pr(X_3=n_{-,3})=q_3(n_{-,3}),\qquad n_{-,3}=0,\ldots,n_3.
\end{aligned}
\label{eq:grouped-marginal-lp}
\end{equation}
This optimization fixes the three marginal distributions and varies the joint distribution to minimize or maximize $F=\Pr(X_1=0,X_2=0,X_3=0)$. The same bounds can be obtained by retaining only each group's zero-count probability $\Pr(X_c = 0)$. Then, we define $\widetilde X_c=0$ when $X_c=0$ and $\widetilde X_c=1$ when $X_c>0$, the fidelity bounds are therefore given by
\begin{equation}
\begin{aligned}
F_{\min/\max}&={\min/\max}\ \Pr(\widetilde X_1=0,\widetilde X_2=0,\widetilde X_3=0),\\
\mathrm{s.t.}\quad
&\Pr(\widetilde X_1=0)=q_1(0)=\frac{1}{2^{n_1}}\sum_{k=0}^{n_1}\binom{n_1}{k}\langle S_1^{(k)}\rangle_\sigma,\\
&\Pr(\widetilde X_2=0)=q_2(0)=\frac{1}{2^{n_2}}\sum_{k=0}^{n_2}\binom{n_2}{k}\langle S_2^{(k)}\rangle_\sigma,\\
&\Pr(\widetilde X_3=0)=q_3(0)=\frac{1}{2^{n_3}}\sum_{k=0}^{n_3}\binom{n_3}{k}\langle S_3^{(k)}\rangle_\sigma,
\end{aligned}
\label{eq:grouped-reduced-lp}
\end{equation}
where the right-hand side of each constraint is given by the expectation value of the projector $ \prod_{i\in A_{c}}{\mathbb{I} + g_i\over 2}$. 
For $\chi$ groups, this reduces $\prod_{c=1}^{\chi}(n_c+1)$ variables to $2^\chi$ without changing either optimal bound.

For hypergraph states, the coloring protocol of Ref.~\cite{prapplied_2019_hypergraph} uses the same group projectors: each test measures $X$ on a color class $A_c$ and $Z$ on its complement, with passing probability $q_c(0)$. 
This coloring protocol is designed to efficiently verify that the prepared state is close to the target state, and it also provides a lower bound on the fidelity. However, this bound can be loose when the prepared state deviates from the target. 
For instance, consider the equally weighted tests, its fidelity bound is
\begin{equation}
F \geq \max\!\left\{0,\sum_{c=1}^{\chi}q_c(0)-\chi+1\right\}.
\label{eq:coloring-fidelity-bounds}
\end{equation}
Under global depolarizing noise, $q_c(0)=1-\eta+\eta 2^{-n_c}$. For the $n=10$-qubit 1D order-$3$ cluster state with $(n_1,n_2,n_3)=(4,3,3)$ and $\eta=0.4$, Eq.~\eqref{eq:coloring-fidelity-bounds} gives a trivial lower bound of zero, whereas the true fidelity is $F=0.6+0.4\times2^{-10}\simeq0.6004$.

\section{Sample complexity of hypergraph states\label{app:large-deviation bounds}}

In practice, stabilizer moments are estimated from finitely many sampled operators and measurement shots, so finite-sample effects must be taken into account. 
In this section, we analyze the sample complexity of estimating stabilizer moments for hypergraph states. We first derive a general bound for general hypergraph states and then specialize it to the 1D order-3 cluster state, showing how its finite interaction range reduces the measurement cost. 

\subsection{Sample complexity for general hypergraph states}

\begin{lemma}
\label{lem:setting_term_parseval}
Given a hypergraph state $|H\rangle=U_H|+\rangle^{\otimes n}$ with stabilizer generators $\{g_i\}_{i=1}^n$, where $U_H\coloneqq\prod_{e\in E}C_e$,
let $I=\{i_1,\ldots,i_k\}\subseteq[n]$ and $g_I\coloneqq\prod_{i\in I}g_i$. Order the computational basis as $|\boldsymbol{x},\boldsymbol{y}\rangle$, with $\boldsymbol{x}$ labeling the qubits in $I$ and $\boldsymbol{y}$ in its complement $I^c = [n]\backslash I$. 
Then
\begin{equation}
\begin{aligned}
    g_I = \sum_{\substack{\boldsymbol{r}\in\mathbb{Z}_2^k\\
    |\boldsymbol{r}|\,\mathrm{even}}}g_{I,\boldsymbol{r}}, \qquad \text{with } ~~
    g_{I,\boldsymbol{r}}\coloneqq \left[
    X_{i_1}Z_{i_1}^{r_1} \otimes \cdots \otimes X_{i_k}Z_{i_k}^{r_k} \right]_I
    \otimes
    \left[\sum_{\boldsymbol{y}\in\mathbb{Z}_2^{n-k}}
    w_{\boldsymbol{r}}(\boldsymbol{y})
    |\boldsymbol{y}\rangle\langle\boldsymbol{y}|\right]_{I^c},
\end{aligned}
\end{equation}
where $w_{\boldsymbol{r}}(\boldsymbol{y}) \coloneqq2^{-k}\sum_{\boldsymbol{x}\in\mathbb{Z}_2^k}h(\boldsymbol{x},\boldsymbol{y})h(\boldsymbol{x}+\boldsymbol{1},\boldsymbol{y})(-1)^{\boldsymbol{r}\cdot\boldsymbol{x}}$, and the function $h(\boldsymbol{x},\boldsymbol{y})\in\{\pm1\}$ is defined by $U_H|\boldsymbol{x},\boldsymbol{y}\rangle=h(\boldsymbol{x},\boldsymbol{y})|\boldsymbol{x},\boldsymbol{y}\rangle$. 
Therefore, each $g_{I,\boldsymbol{r}}$ with even $|\boldsymbol{r}|$ is Hermitian and accessible from one local Pauli setting, in which all qubits in $I^c$ are measured in the $Z$-basis. 
For a single shot measurement, let $p_{\boldsymbol{r}}\in\{\pm1\}$ be the product of these local $X$- and $Y$-measurement outcomes on $I$ and let $\boldsymbol{y}$ be the outcome bit-string from qubits in $I^c$. The estimator 
\begin{equation}
    A_{I,\boldsymbol{r}}
    \coloneqq(-1)^{|\boldsymbol{r}|/2}
    p_{\boldsymbol{r}}w_{\boldsymbol{r}}(\boldsymbol{y})\in[-1,1],
\label{eq:single-shot-setting-estimator}
\end{equation}
is unbiased, i.e., $\mathbb{E}[A_{I,\boldsymbol{r}}]=\langle g_{I,\boldsymbol{r}}\rangle_\sigma$. Moreover, $\sum_{\substack{\boldsymbol{r}\in\mathbb{Z}_2^k\\|\boldsymbol{r}|\,\mathrm{even}}}g_{I,\boldsymbol{r}}^{2}=\mathbb{I}$.
\end{lemma}

\begin{proof}
For the sake of clarity, we relabel the qubits so that $I=\{1,\ldots,k\}$, and the general case follows identically. With $U_H$ and $h$ as defined in the lemma, we have $g_I=U_H(X_1\otimes\cdots\otimes X_k)U_H^\dagger$. Since $X_1\otimes\cdots\otimes X_k$ maps $|\boldsymbol{x},\boldsymbol{y}\rangle$ to $|\boldsymbol{x}+\boldsymbol{1},\boldsymbol{y}\rangle$, we have $g_I|\boldsymbol{x},\boldsymbol{y}\rangle=h(\boldsymbol{x},\boldsymbol{y})h(\boldsymbol{x}+\boldsymbol{1},\boldsymbol{y})|\boldsymbol{x}+\boldsymbol{1},\boldsymbol{y}\rangle$. Then, we can first decompose $g_I$ in the computational basis as
\begin{equation}
    \label{eq:B3}
    g_I =\sum_{\boldsymbol{x},\boldsymbol{y}}
    d_I(\boldsymbol{x},\boldsymbol{y})
    |\boldsymbol{x}+\boldsymbol{1},\boldsymbol{y}\rangle\langle\boldsymbol{x},\boldsymbol{y}|
    =(X_1\otimes \cdots \otimes X_k) \sum_{\boldsymbol{x},\boldsymbol{y}}
    d_I(\boldsymbol{x},\boldsymbol{y})
    |\boldsymbol{x},\boldsymbol{y}\rangle\langle\boldsymbol{x},\boldsymbol{y}|,
    \qquad \text{with} \quad 
    d_I(\boldsymbol{x},\boldsymbol{y})
    \coloneqq h(\boldsymbol{x},\boldsymbol{y})h(\boldsymbol{x}+\boldsymbol{1},\boldsymbol{y})\in\{\pm1\}.
\end{equation}
For fixed $\boldsymbol{y}$, $d_I(\boldsymbol{x},\boldsymbol{y})$ is actually the Fourier transform of $w_{\boldsymbol r}(\boldsymbol y)$ defined in the lemma, i.e.,
\begin{equation}
    d_I(\boldsymbol{x},\boldsymbol{y})
    =\sum_{\boldsymbol{r}\in\mathbb{Z}_2^k}w_{\boldsymbol{r}}(\boldsymbol{y}) (-1)^{\boldsymbol{r}\cdot\boldsymbol{x}},
    \qquad w_{\boldsymbol{r}}(\boldsymbol{y}) = 2^{-k}\sum_{\boldsymbol{x}\in\mathbb{Z}_2^k}d_I(\boldsymbol{x},\boldsymbol{y})(-1)^{\boldsymbol{r}\cdot\boldsymbol{x}}. 
\end{equation}
Substituting such a Fourier expansion into Eq.~\eqref{eq:B3} and using $\sum_{\boldsymbol{z}}(-1)^{\boldsymbol{r}\cdot\boldsymbol{z}}|\boldsymbol{z}\rangle\langle\boldsymbol{z}|=Z_1^{r_1}\otimes\cdots\otimes Z_k^{r_k}$ gives
\begin{equation}
    \begin{aligned}
    g_I
    = (X_1\otimes \cdots \otimes X_k) \sum_{\boldsymbol{y},\boldsymbol{r}}
    w_{\boldsymbol{r}}(\boldsymbol{y}) \sum_{\boldsymbol{x}}(-1)^{\boldsymbol{r}\cdot \boldsymbol{x}}|\boldsymbol{x},\boldsymbol{y}\rangle\langle\boldsymbol{x},\boldsymbol{y}| =\sum_{\boldsymbol{r}\in\mathbb{Z}_2^k}
    \left[(X_1Z_1^{r_1})\otimes\cdots\otimes
    (X_kZ_k^{r_k})\right]
    \otimes
    \left[\sum_{\boldsymbol{y}\in \mathbb{Z}_2^{n-k}}w_{\boldsymbol{r}}(\boldsymbol{y})
    |\boldsymbol{y}\rangle\langle\boldsymbol{y}|\right] =\sum_{\substack{\boldsymbol{r}\in\mathbb{Z}_2^k\\
    |\boldsymbol{r}|\,\mathrm{even}}}g_{I,\boldsymbol{r}}.
    \end{aligned}
\end{equation}
The requirement on even $|\boldsymbol r|$ in the third equality follows from the fact that $d_I(\boldsymbol{x}+\boldsymbol{1},\boldsymbol{y})=d_I(\boldsymbol{x},\boldsymbol{y})$, which implies $w_{\boldsymbol{r}}(\boldsymbol{y})=(-1)^{|\boldsymbol{r}|}w_{\boldsymbol{r}}(\boldsymbol{y})$ and hence removes all odd-weight terms. 
Measuring the corresponding $X$ or $Y$ observables on $I$ and $Z$ on $I^c$ produces the single-shot estimator $A_{I,\boldsymbol{r}}$ defined in the lemma. Since $|w_{\boldsymbol{r}}(\boldsymbol{y})|\leq2^{-k}\sum_{\boldsymbol{x}}|d_I(\boldsymbol{x},\boldsymbol{y})|=1$ , we have $A_{I,\boldsymbol{r}}\in[-1,1]$. 

Finally, Parseval's identity and $d_I(\boldsymbol{x},\boldsymbol{y})^2=1$ give $\sum_{\boldsymbol{r}}w_{\boldsymbol{r}}(\boldsymbol{y})^2 = 2^{-k}\sum_{\boldsymbol{x}\in\mathbb{Z}_2^k} d_I(\boldsymbol{x},\boldsymbol{y})^2=1$ for every $\boldsymbol{y}$. 
For even $|\boldsymbol{r}|$, we obtain
\begin{equation}
\begin{aligned}
    \sum_{\substack{\boldsymbol{r}\in\mathbb{Z}_2^k\\
    |\boldsymbol{r}|\,\mathrm{even}}}
    g_{I,\boldsymbol{r}}^{2}
    &=\mathbb{I}_I\otimes
    \sum_{\boldsymbol{y}}
    \sum_{\substack{\boldsymbol{r}\in\mathbb{Z}_2^k\\
    |\boldsymbol{r}|\,\mathrm{even}}}
    w_{\boldsymbol{r}}(\boldsymbol{y})^2
    |\boldsymbol{y}\rangle\langle\boldsymbol{y}| = \mathbb{I}_I\otimes
    \sum_{\boldsymbol{y}}
    \sum_{\boldsymbol{r}\in\mathbb{Z}_2^k}
    w_{\boldsymbol{r}}(\boldsymbol{y})^2
    |\boldsymbol{y}\rangle\langle\boldsymbol{y}|  =\mathbb{I}. 
\end{aligned}
\end{equation}

\end{proof}

% \begin{lemma}
% \label{lem:bernstein-inequality}
% Let $Y_1,\ldots,Y_N$ be independent, zero-mean random variables satisfying
% $|Y_i| \leq b$ for every $i=1,\ldots,N$.
% Then, for every $\epsilon>0$ and every
% $v \geq \operatorname{Var}\!\left(\sum_{i=1}^N Y_i\right)$, we have
% \begin{equation}
% \Pr\!\left(
%     \left|\sum_{i=1}^N Y_i\right|\geq \epsilon
% \right)
% \leq
% 2\exp\!\left[
%     -\frac{\epsilon^2}{2\left(v+b\epsilon/3\right)}
% \right].
% \label{eq:bernstein-inequality}
% \end{equation}
% \end{lemma}

% \begin{proof}
%     See~\cite{Bernstein_inequality}.
% \end{proof}

For each $k$ and every $k$-element subset $I$, we define $\mathcal R_I\coloneqq\{\boldsymbol r\in\mathbb Z_2^k:|\boldsymbol r|\text{ is even},\ g_{I,\boldsymbol r}\neq0\}$ with $|\mathcal R_I|\leq 2^{k-1}$. Its cardinality $|\mathcal R_I|$ is the number of Pauli settings required to estimate $\langle g_I\rangle_\sigma$. 
By construction, all terms outside $\mathcal R_I$ vanish, so $g_I=\sum_{\boldsymbol r\in\mathcal R_I}g_{I,\boldsymbol r}$ and $\sum_{\boldsymbol r\in\mathcal R_I}g_{I,\boldsymbol r}^2=\mathbb I$. 
We now establish a finite-sample guarantee for jointly estimating the stabilizer moments up to order $K$, together with a bound on the required number of copies.

\begin{theorem}
\label{thm:moment_sampling}
For each $k\in[K]$ with $ K \leq n$, we define $\mu_k\coloneqq\langle S^{(k)}\rangle_\sigma$ and sample $I_1,\ldots,I_l \subseteq [n]$ with $l=\left\lceil8\ln(4K/\delta)/\epsilon^2\right\rceil$ independently and uniformly from the $k$-element subsets of $[n]$. 
For each sampled subset $I_j$, we estimate $\langle g_{I_j}\rangle_\sigma$ using $|\mathcal R_{I_j}|$ local Pauli settings, each measured on $M=\left\lceil8(1+\epsilon/3)\ln(4K/\delta)/(l\epsilon^2)\right\rceil$ independent copies. We then construct the estimator $\tilde{\mu}_k\coloneqq\frac{1}{Ml}\sum_{j=1}^l\sum_{\boldsymbol r\in\mathcal R_{I_j}}\sum_{t=1}^M A_{I_j,\boldsymbol r,t}$, where $A_{I_j,\boldsymbol r,t}$ denotes the single-shot estimator defined in Eq.~\eqref{eq:single-shot-setting-estimator} for the $t$-th copy. Thus, the estimators satisfy $\Pr\!\left[\bigcap_{k = 1}^K|\widetilde\mu_k-\mu_k|\leq\epsilon\right] \geq 1 - \delta$ with a total number of $N_{\rm tot}=M\sum_{k=1}^{K}\sum_{j=1}^l|\mathcal R_{I_j}|$ copies. 
Using $\lceil x\rceil\leq x+1$ in the definitions of $l$ and $M$, the expected number of copies satisfies
\begin{equation}
\mathbb{E}[N_{\rm tot}] = Ml\sum_{k=1}^{K}\mathbb E_{|I|=k}[|\mathcal R_I|] \leq \left[\frac{8(2+\epsilon/3)\ln(4K/\delta)}{\epsilon^2}+1\right](2^K-1) = O\!\left[\frac{2^K\ln(4K/\delta)}{\epsilon^2}\right].
\label{eq:general-mean-sampling-cost}
\end{equation}
\end{theorem}

\begin{proof}
For each $k\in[K]$, suppose that $\langle g_{I_1}\rangle_\sigma, \cdots, \langle g_{I_l}\rangle_\sigma$ with $l = \lceil2\ln(2/\delta_1)/\epsilon_1^2 \rceil$ are known exactly. We construct the estimator $\widehat{\mu}_k\coloneqq\frac{1}{l}\sum_{j=1}^l\langle g_{I_j}\rangle_\sigma$. 
Since the subsets are sampled independently and uniformly, we have $\mathbb{E}[\widehat{\mu}_k]=\mu_k$. Because $\langle g_{I_j}\rangle_\sigma \in[-1,1]$, Hoeffding's inequality directly gives
\begin{equation}
    \label{eq:moment_hoeffding}
    \Pr\left[|\widehat{\mu}_k-\mu_k|\geq\epsilon_1\right]
    \leq2\exp\left(-{l\epsilon_1^2}/2\right)
    \leq\delta_1.
\end{equation}

We now consider the estimation error of each $\langle g_{I_j}\rangle_\sigma$. 
For each $\boldsymbol r\in\mathcal R_{I_j}$, we perform the corresponding Pauli measurement on $M = \lceil 2(1+2\epsilon_2/3)\ln(2/\delta_2)/(l\epsilon_2^2) \rceil$ independent copies. Then, we construct the estimator $\tilde{\mu}_k\coloneqq\frac{1}{Ml}\sum_{j=1}^l\sum_{\boldsymbol r\in\mathcal R_{I_j}}\sum_{t=1}^M A_{I_j,\boldsymbol r,t}$ that is unbiased $\mathbb{E}[\tilde{\mu}_k]=\widehat{\mu}_k$. Define $Y_{j\boldsymbol{r}t}\coloneqq(A_{I_j,\boldsymbol{r},t}-\mathbb{E}[A_{I_j,\boldsymbol{r},t}])/(Ml)$ such that $|Y_{j\boldsymbol{r}t}|\leq2/(Ml)$. Lemma~\ref{lem:setting_term_parseval} bounds the total variance as
\begin{equation}
\begin{aligned}
    \sum_{j,\boldsymbol{r},t}\operatorname{Var}(Y_{j\boldsymbol{r}t})
    =\frac{1}{M^2l^2} \sum_{j,\boldsymbol{r},t} \left( \mathbb{E}[A_{I_j,\boldsymbol{r},t}^2] - \mathbb{E}[A_{I_j,\boldsymbol{r},t}]^2\right) \leq\frac{1}{M^2l^2}
    \sum_{j,\boldsymbol{r},t}\mathbb{E}[A_{I_j,\boldsymbol{r},t}^2] = \frac{1}{Ml^2}\sum_{j=1}^l \sum_{\boldsymbol{r}\in\mathcal R_{I_j}}
    \left\langle g_{I_j,\boldsymbol{r}}^2\right\rangle_\sigma
    =\frac{1}{Ml}.
\end{aligned}
\end{equation}
Bernstein's inequality therefore yields
\begin{equation}
    \label{eq:moment_bernstein}
    \Pr\left[|\tilde{\mu}_k-\widehat{\mu}_k|
    \geq\epsilon_2\right]
    \leq2\exp\left[
    -\frac{Ml\epsilon_2^2}{2(1+2\epsilon_2/3)}
    \right]
    \leq\delta_2.
\end{equation}
The triangle inequality $|\tilde{\mu}_k-\mu_k|\leq|\tilde{\mu}_k-\widehat{\mu}_k|+|\widehat{\mu}_k-\mu_k|$ implies that an error of at least $\epsilon_1+\epsilon_2$ can occur only if one of the two contributions exceeds its respective threshold. Hence, the union bound and Eqs.~\eqref{eq:moment_hoeffding} and~\eqref{eq:moment_bernstein} give
\begin{equation}
\begin{aligned}
    \Pr\left[|\tilde{\mu}_k-\mu_k|\geq\epsilon_1+\epsilon_2\right]
    &\leq
    \Pr\left[|\tilde{\mu}_k-\widehat{\mu}_k|\geq\epsilon_2\right]
    +\Pr\left[|\widehat{\mu}_k-\mu_k|\geq\epsilon_1\right]
    \leq\delta_2+\delta_1.
\end{aligned}
\end{equation}
Thus, let $\epsilon_1 = \epsilon_2 = \epsilon/2$, and $\delta_1 = \delta_2 = \delta/2K$, each order-$k$ estimate fails with probability at most $\delta/K$. Applying the union bound gives 
\begin{equation}
\Pr\!\left[\bigcap_{k=1}^K|\widetilde\mu_k-\mu_k|\leq\epsilon\right] \geq 1 -  \Pr\!\left[\bigcup_{k=1}^K|\widetilde\mu_k-\mu_k|\geq\epsilon\right]
\geq 1 - \sum_{k=1}^{K}\Pr\!\left[|\widetilde\mu_k-\mu_k|\geq\epsilon\right]
\geq 1 - \sum_{k=1}^{K}\frac{\delta}{K}
=1 - \delta.
\end{equation}
In total, the protocol uses $N_{\rm tot}\coloneqq M\sum_{k=1}^{K}\sum_{j=1}^{l}|\mathcal R_{I_j}|$ copies. Taking the expectation over the sampled subsets yields
\begin{equation}
\begin{aligned}
\mathbb E[N_{\rm tot}] = Ml\sum_{k=1}^{K}\mathbb E_{|I|=k}[|\mathcal R_I|] \leq Ml\sum_{k=1}^{K}2^{k-1}=Ml(2^K-1).
\end{aligned}
\end{equation}
which concludes the proof.
\end{proof}

\subsection{Sample complexity for the 1D order-3 cluster state}
\label{app:collision-adaptive-moments}

Theorem~\ref{thm:moment_sampling} expresses the sampling cost using $\mathbb E_{|I|=k}[|\mathcal R_I|] \leq 2^{k - 1}$. For the 1D order-3 cluster state as shown in Fig.~\ref{fig:supp_1Dcluster}, the finite interaction range allows this support to be bounded by $\mathbb E_{|I|=k}[|\mathcal R_I|] \leq (1+{2k\over n})^{k - 1}$. For example, the product $g_1g_4\cdots g_n$ contains approximately $k=n/3$ generators but only one term $g_{I,\boldsymbol{r}}$ appears in its Pauli decomposition.

\begin{figure}[t]
    \centering
    \includegraphics[width=0.6\linewidth]{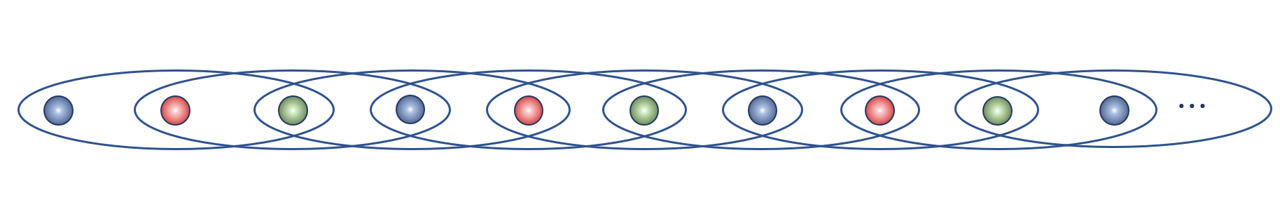}
    \caption{\label{fig:supp_1Dcluster}Structure of the 1D order-3 cluster state $|C\rangle$. Each vertex represents a qubit, each loop denotes a three-qubit hyperedge $\{i,i+1,i+2\}$, and the vertex colors indicate a 3-coloring.}
\end{figure}

\begin{lemma}
\label{lem:nearby-generator-collision}
Let $I=\{i_1<\cdots<i_k\}$ be a uniformly sampled $k$-element subset of $[n]$. Define $a(I)\coloneqq\sum_{j=1}^{k-1}\mathds{1}[i_{j+1}-i_j\leq2]$, which counts how many neighboring pairs $i_j,i_{j+1}$ in the ordered list are separated by at most two sites. Then
\begin{equation}
\mathbb E_{|I|=k}[2^{a(I)}]
\leq
\left(1+\frac{2k}{n}\right)^{k-1}.
\label{eq:mean-setting-factor-bound}
\end{equation}
\end{lemma}

\begin{proof}
The expectation $\mathbb E_{|I|=k}[2^{a(I)}]$ depends only on how often consecutive selected indices are separated by at most two sites. We encode these separations by the gaps $\upDelta_0=i_1-1$, $\upDelta_j=i_{j+1}-i_j-1$ for $1\leq j\leq k-1$, and $\upDelta_k=n-i_k$, as illustrated by
\begin{equation}
\underbrace{1,2,\ldots,i_1-1}_{\upDelta_0},\
\boxed{i_1},\
\underbrace{i_1+1,\ldots,i_2-1}_{\upDelta_1},\
\boxed{i_2},\ \ldots,\
\boxed{i_j},\
\underbrace{i_j+1,\ldots,i_{j+1}-1}_{\upDelta_j},\
\boxed{i_{j+1}},\ \ldots,\
\boxed{i_k},\
\underbrace{i_k+1,\ldots,n}_{\upDelta_k}.
\end{equation}
The gaps are nonnegative integers satisfying $\sum_{j=0}^{k}\upDelta_j=n-k$. Conversely, any such tuple $(\upDelta_0,\ldots,\upDelta_k)$ uniquely determines the sampled subset through $i_1=\upDelta_0+1$ and so on. Hence, sampling $I$ uniformly induces a uniform sample over the $\binom nk$ admissible gap tuples $(\upDelta_0,\ldots,\upDelta_k)$, i.e., they satisfy $\sum_{j = 0}^k \upDelta_j = n - k$.

Because $a(I)=\sum_{j=1}^{k-1}\mathds{1}[\upDelta_j\leq1]$, we expand
$2^{a(I)}=\prod_{j=1}^{k-1}(1+\mathds{1}[\upDelta_j\leq1])$. Hence,
\begin{equation}
\begin{aligned}
\mathbb{E}_{|I|=k}[2^{a(I)}] = \frac{1}{\binom nk}\sum_{\text{admissible } (\upDelta_0,\ldots,\upDelta_k)}\prod_{j=1}^{k-1}(1+\mathds{1}[\upDelta_j\leq1]) = \sum_{R\subseteq\{\upDelta_1,\ldots,\upDelta_{k-1}\}}
\Pr[\upDelta\leq1\ \text{for every }\upDelta\in R]. 
\end{aligned}
\end{equation}
For a fixed $R$, suppose that $q$ of its gaps equal one and the remaining $|R|-q$ gaps equal zero. After choosing these $q$ gaps, the remaining $k+1-|R|$ gaps are nonnegative integers with total $n-k-q$. The number of such assignments is
$\binom{(n-k-q)+(k+1-|R|)-1}{(k+1-|R|)-1}=\binom{n-|R|-q}{k-|R|}$. Therefore,
\begin{equation}
\begin{aligned}
\Pr[\upDelta\leq1\ \text{for every }\upDelta\in R]
&=\frac{1}{\binom nk}\sum_{q=0}^{|R|}\binom{|R|}{q}\binom{n-|R|-q}{k-|R|}
&\leq2^{|R|}\frac{\binom{n-|R|}{k-|R|}}{\binom nk}
=2^{|R|}\frac{k!(n-|R|)!}{(k-|R|)!n!}
&=2^{|R|}\prod_{j=0}^{|R|-1}\frac{k-j}{n-j}
\leq\left(\frac{2k}{n}\right)^{|R|}.
\end{aligned}
\label{eq:short-gap-probability}
\end{equation}
The first inequality uses $\binom{n-|R|-q}{k-|R|}\leq\binom{n-|R|}{k-|R|}$ and $\sum_{q=0}^{|R|}\binom{|R|}{q}=2^{|R|}$, and the last uses $(k-j)/(n-j)\leq k/n$.

Applying Eq.~\eqref{eq:short-gap-probability} to the probability sum above gives
\begin{equation}
\mathbb E_{|I|=k}[2^{a(I)}]
\leq
\sum_{R\subseteq\{\upDelta_1,\ldots,\upDelta_{k-1}\}}
\left(\frac{2k}{n}\right)^{|R|}
=\left(1+\frac{2k}{n}\right)^{k-1}.
\end{equation}
\end{proof}

We now connect $2^{a(I)}$ to the measurement cost using the fact that only the product of nearby generators can increase the number of required settings.

\begin{lemma}
\label{lem:component-setting-bound}
Write $I=\{i_1<\cdots<i_k\}$ and split it into $c(I)$ blocks $J_1,\ldots,J_{c(I)}$, starting a new block whenever $i_{j+1}-i_j>2$. For example, $\{1,3,4,8,10\}$ is split into $\{1,3,4\}$ and $\{8,10\}$. By construction, $c(I)=k-a(I)$, where $a(I)$ is defined in Lemma~\ref{lem:nearby-generator-collision}.
With $\mathcal R_I$ as defined above, let $\mathcal Q_I$ be the set of labels $\boldsymbol{r}=(r_1,\ldots,r_k)$ having even parity within every block $J_\nu$.
Then, for the 1D order-3 cluster state, we have $\mathcal R_I\subseteq\mathcal Q_I$ and
\begin{equation}
|\mathcal R_I|
\leq|\mathcal Q_I|
=2^{a(I)}.
\label{eq:component-setting-bound}
\end{equation}
\end{lemma}

\begin{proof}
We first prove that $\mathcal R_I\subseteq\mathcal Q_I$. 
Order the computational basis as $|\boldsymbol{x},\boldsymbol{y}\rangle$, where $\boldsymbol{x}\in\mathbb Z_2^k$ gives the computational-basis bits on $I$ and $\boldsymbol{y}$ gives those on $I^c$. 
For the 1D order-3 cluster state, let $\boldsymbol{z}=(\boldsymbol{x},\boldsymbol{y})$ denote the full computational-basis string in physical-site order. The phase function entering Eq.~\eqref{eq:B3} is $h(\boldsymbol{x},\boldsymbol{y})=\prod_{q=1}^{n-2}(-1)^{z_qz_{q+1}z_{q+2}}$. For a fixed block $J_\nu$, every hyperedge $\{q,q+1,q+2\}$ contains selected qubits from at most one block, so we can flip all bits in $J_\nu$ without changing the phase function:
\begin{equation}
d_I(\boldsymbol{x},\boldsymbol{y})
=h(\boldsymbol{x},\boldsymbol{y})h(\boldsymbol{x}+\boldsymbol{1},\boldsymbol{y})
=\prod_{q=1}^{n-2}(-1)^{z_qz_{q+1}z_{q+2}+(z_q+\mathds{1}[q\in I])(z_{q+1}+\mathds{1}[q+1\in I])(z_{q+2}+\mathds{1}[q+2\in I])}
=d_I(\boldsymbol{x}',\boldsymbol{y}),
\qquad x_j' \coloneqq
\begin{cases}
x_j+1,&i_j\in J_\nu,\\
x_j,&i_j\notin J_\nu.
\end{cases}
\end{equation}

Moreover, $\boldsymbol{x}'$ differs from $\boldsymbol{x}$ only on the entries associated with $J_\nu$. Therefore, over $\mathbb Z_2$, we have
\begin{equation}
\begin{aligned}
\boldsymbol{r}\cdot\boldsymbol{x}'
&=
\sum_{j:\,i_j\notin J_\nu}r_jx_j'
+\sum_{j:\,i_j\in J_\nu}r_jx_j'
&=
\sum_{j:\,i_j\notin J_\nu}r_jx_j
+\sum_{j:\,i_j\in J_\nu}r_j(x_j+1)
&=
\sum_{j=1}^{k}r_jx_j
+\sum_{j:\,i_j\in J_\nu}r_j
&=
\boldsymbol{r}\cdot\boldsymbol{x}
+\sum_{j:\,i_j\in J_\nu}r_j
\pmod2.
\end{aligned}
\end{equation}
Thus, the block-wise flip changes the Walsh phase by the factor
$(-1)^{\sum_{j:\,i_j\in J_\nu}r_j}$:
\begin{equation}
\begin{aligned}
w_{\boldsymbol{r}}(\boldsymbol{y})
=2^{-k}\sum_{\boldsymbol{x}}
d_I(\boldsymbol{x},\boldsymbol{y})
(-1)^{\boldsymbol{r}\cdot\boldsymbol{x}}
=2^{-k}\sum_{\boldsymbol{x}}
d_I(\boldsymbol{x}',\boldsymbol{y})
(-1)^{\boldsymbol{r}\cdot\boldsymbol{x}'}
=(-1)^{\sum_{j:\,i_j\in J_\nu}r_j}
w_{\boldsymbol{r}}(\boldsymbol{y}).
\end{aligned}
\end{equation}
Thus, for any $\boldsymbol{y}$, $w_{\boldsymbol{r}}(\boldsymbol{y})$ can be nonzero only if $\sum_{j:\,i_j\in J_\nu}r_j=0\pmod2$. Applying this condition to every block proves $\mathcal R_I\subseteq\mathcal Q_I$.

For the cardinality of $\mathcal Q_I$, all but one of the bits $r_j$ in each block $J_\nu$ can be chosen freely. Thus, we have
\begin{equation}
\begin{aligned}
|\mathcal Q_I|
&=\prod_{\nu=1}^{c(I)}2^{|J_\nu|-1}
=2^{\sum_{\nu=1}^{c(I)}|J_\nu|-c(I)}
&=2^{k-c(I)}
=2^{a(I)},
\end{aligned}
\end{equation}
which proves Eq.~\eqref{eq:component-setting-bound}.
\end{proof}

\begin{corollary}
\label{cor:collision-adaptive-moment-sampling}
For the 1D order-3 cluster state, the estimator in Theorem~\ref{thm:moment_sampling} obeys the same simultaneous confidence guarantee, and its expected total cost satisfies
\begin{equation}
\mathbb E[N_{\rm tot}]
=Ml\sum_{k=1}^{K}\mathbb E_{|I|=k}[|\mathcal R_I|]
\leq Ml\min\!\left\{2^K-1,\frac{n}{2K}\left[\left(1+\frac{2K}{n}\right)^K-1\right]\right\}.
\label{eq:collision-adaptive-total-cost}
\end{equation}
\end{corollary}

\begin{proof}
The two preceding lemmas imply
$\mathbb E_{|I|=k}[|\mathcal R_I|]\leq\mathbb E_{|I|=k}[2^{a(I)}]\leq(1+2k/n)^{k-1}$.
Since $k\leq K$, summing the resulting geometric bound gives
\begin{equation}
\sum_{k=1}^{K}\mathbb E_{|I|=k}[|\mathcal R_I|]
\leq\sum_{k=1}^{K}\left(1+\frac{2K}{n}\right)^{k-1}
=\frac{n}{2K}\left[\left(1+\frac{2K}{n}\right)^K-1\right],
\end{equation}
which completes the proof.
\end{proof}

\section{Simulation of statistics for the 1D order-3 cluster state}
\label{sec:supp-local-depolarizing}
We now derive how to simulate the stabilizer moments for the 1D order-3 cluster state. Under local depolarizing noise, these moments are evaluated using a fixed-dimensional transfer matrix. Under global depolarizing noise, they admit direct expressions. 

For each sampled $I$, we use all local Pauli settings indexed by $\mathcal Q_I$, which cover both candidate targets considered below. Accordingly, the expected measurement cost in Fig.~\ref{fig:topology_certification} is evaluated using the upper bound in Eq.~\eqref{eq:collision-adaptive-total-cost}.

\subsection{Notation for the 1D order-3 cluster state}
\textbf{Experimentally prepared state.} We take the 1D order-3 cluster state as the ideal prepared state and consider its noisy realizations under local and global depolarizing noise. 
For $n$ qubits, let
\begin{equation}
E_3\coloneqq\bigl\{\{i,i+1,i+2\}:i=1,\ldots,n-2\bigr\}.
\label{eq:rank-adaptive-E3}
\end{equation}
The state is defined as
\begin{equation}
|C\rangle\coloneqq\prod_{e\in E_3}C_e|+\rangle^{\otimes n}=2^{-n/2}\sum_{\boldsymbol{v}\in\mathbb Z_2^n}(-1)^{f(\boldsymbol{v})}|\boldsymbol{v}\rangle, \qquad  f(\boldsymbol{v})\coloneqq\sum_{i=1}^{n-2}v_iv_{i+1}v_{i+2}\pmod 2.
\label{eq:rank-adaptive-C3}
\end{equation}
Let $\mathcal D\coloneqq\mathcal D_1^{\otimes n}$ denote the local depolarizing channel for $n$ qubits, where the local depolarizing channel for a single qubit is $\mathcal D_1(\rho)\coloneqq(1-p)\rho+\frac{p}{3}\sum_{P\in\{X,Y,Z\}}P\rho P$. 
Under local and global depolarizing noise, respectively, the prepared states are
\begin{equation}
\sigma_l\coloneqq\mathcal D\!\left(|C\rangle\langle C|\right),\qquad \sigma_g\coloneqq(1-\eta)|C\rangle\langle C|+\frac{\eta}{2^n}\mathbb I.
\label{eq:appD3}
\end{equation}

\textbf{Candidate target states.} %The local Pauli measurement outcomes collected for the connected target $|C\rangle$ can also be used to estimate the stabilizers of any hypergraph state with edge set $E\subseteq E_3$. Thus, a common statistics can be used to bound the fidelity of multiple target states. 
For clarity, we consider two targets: the connected chain $|C\rangle$ and, when $n$ is divisible by four, the disconnected chain $|C_\otimes\rangle\coloneqq\bigl(|C^{(n/4)}\rangle\bigr)^{\otimes4}$ on four consecutive blocks. Let $U$ and $U_\otimes$ denote their preparation unitaries, i.e., $|C\rangle = U|+\rangle^{\otimes n}$ and $|C_\otimes\rangle = U_\otimes|+\rangle^{\otimes n}$, and define
\begin{equation}
g_i\coloneqq UX_iU^\dagger,\qquad s_{\boldsymbol{x}}\coloneqq\prod_{i=1}^n g_i^{x_i},\qquad g_i^\otimes\coloneqq U_\otimes X_iU_\otimes^\dagger,\qquad s_{\boldsymbol{x}}^\otimes\coloneqq\prod_{i=1}^n(g_i^\otimes)^{x_i}.
\label{eq:candidate-stabilizing-operators}
\end{equation}
For $k\in[n]$, define $S^{(k)}\coloneqq L_k^{-1}\sum_{|\boldsymbol{x}|=k}s_{\boldsymbol{x}}$ and $S_\otimes^{(k)}\coloneqq L_k^{-1}\sum_{|\boldsymbol{x}|=k}s_{\boldsymbol{x}}^\otimes$.

\subsection{Statistics under local depolarizing noise}

Using $\operatorname{Tr}\!\left[O\mathcal D\!\left(|C\rangle\langle C|\right)\right]=\operatorname{Tr}\!\left[\mathcal D^\dagger(O)|C\rangle\langle C|\right]$, we evaluate the noisy expectation values in the Heisenberg picture.
Since every Hermitian operator $O$ admits a Pauli expansion, we first determine the action of $\mathcal D^\dagger$ on a Pauli string. 
Let $\boldsymbol{x},\boldsymbol{z}\in\mathbb{Z}_2^n$ and define the Pauli string $P_{\boldsymbol{x},\boldsymbol{z}} \coloneqq i^{\boldsymbol{x}\cdot\boldsymbol{z}}X^{\boldsymbol{x}}Z^{\boldsymbol{z}}=\bigotimes_{i=1}^nP_{x_i,z_i}$, where $P_{x_i,z_i}\coloneqq i^{x_i z_i}X_i^{x_i}Z_i^{z_i}\in\{\mathbb I_i,X_i,Y_i,Z_i\}$. In the phase $i^{\boldsymbol{x}\cdot\boldsymbol{z}}$, the dot product $\boldsymbol{x}\cdot\boldsymbol{z}=\sum_i x_i z_i$ is evaluated over the integers rather than over $\mathbb{Z}_2$. Since $\mathcal D$ acts independently on all qubits, we have
\begin{equation}
\begin{aligned}
\mathcal D^{\dagger}(P_{\boldsymbol{x},\boldsymbol{z}})
=\bigotimes_{i=1}^n\mathcal D_1^{\dagger}(P_{x_i,z_i})=\bigotimes_{i=1}^n\left[(1-p)P_{x_i,z_i}
+\frac{p}{3}\sum_{Q\in\{X,Y,Z\}}QP_{x_i,z_i}Q\right]=\lambda^{|\boldsymbol{x}\vee\boldsymbol{z}|}P_{\boldsymbol{x},\boldsymbol{z}},
\end{aligned}
\end{equation}
where $\lambda\coloneqq1-4p/3$. 
The last equality follows from $\sum_{Q\in\{X,Y,Z\}}QP_{x_i,z_i}Q=-P_{x_i,z_i}$ unless $P_{x_i,z_i}=\mathbb I$. 
Thus, $|\boldsymbol{x}\vee\boldsymbol{z}|=\sum_{i=1}^n\mathds{1}[(x_i,z_i)\neq(0,0)]$ counts the Pauli weight of $P_{\boldsymbol{x},\boldsymbol{z}}$. 
Next, we express the stabilizing operators in the Pauli basis.

\begin{lemma}
\label{lem:S-pauli-expansion}
For $\boldsymbol{x}\in\mathbb Z_2^n$, the stabilizers $s_{\boldsymbol{x}}=\sum_{\boldsymbol{z}}\alpha_{\boldsymbol{x}}(\boldsymbol{z})P_{\boldsymbol{x},\boldsymbol{z}}$ and $s_{\boldsymbol{x}}^\otimes=\sum_{\boldsymbol{z}}\alpha_{\boldsymbol{x}}^\otimes(\boldsymbol{z})P_{\boldsymbol{x},\boldsymbol{z}}$ have Pauli support only on $\{P_{\boldsymbol{x},\boldsymbol{z}}:\boldsymbol{z}\in\mathbb Z_2^n\}$. These coefficients are given by $\alpha_{\boldsymbol{x}}(\boldsymbol{z})=\langle C|P_{\boldsymbol{x},\boldsymbol{z}}|C\rangle$ and $\alpha_{\boldsymbol{x}}^\otimes(\boldsymbol{z})=\langle C_\otimes|P_{\boldsymbol{x},\boldsymbol{z}}|C_\otimes\rangle$. 
Consequently, their expectation values on the locally depolarized connected state are
\begin{equation}
\begin{aligned}
\langle s_{\boldsymbol{x}}\rangle_{\sigma_l}=\langle C|\mathcal D^\dagger(s_{\boldsymbol{x}})|C\rangle=\sum_{\boldsymbol{z}\in\mathbb Z_2^n}\lambda^{|\boldsymbol{x}\vee\boldsymbol{z}|}\left|\alpha_{\boldsymbol{x}}(\boldsymbol{z})\right|^2, \qquad
\langle s_{\boldsymbol{x}}^\otimes\rangle_{\sigma_l}=\langle C|\mathcal D^\dagger(s_{\boldsymbol{x}}^\otimes)|C\rangle=\sum_{\boldsymbol{z}\in\mathbb Z_2^n}\lambda^{|\boldsymbol{x}\vee\boldsymbol{z}|}\alpha_{\boldsymbol{x}}^\otimes(\boldsymbol{z})\alpha_{\boldsymbol{x}}(\boldsymbol{z}).
\end{aligned}
\label{eq:basic-noise-sum}
\end{equation}
\end{lemma}

\begin{proof}
We first focus on the coefficient for the disconnected target, and then the connected case follows identically. For a fixed $\boldsymbol{x}$, the diagonality of $U_\otimes$ implies that $s_{\boldsymbol{x}}^\otimes=U_\otimes X^{\boldsymbol{x}}U_\otimes^\dagger$ is supported only on computational basis matrix elements of the form $|\boldsymbol{u}+\boldsymbol{x}\rangle\langle\boldsymbol{u}|$, i.e., $\langle\boldsymbol{v}|s_{\boldsymbol{x}}^\otimes|\boldsymbol{u}\rangle=0$ whenever $\boldsymbol{v}\neq\boldsymbol{u}+\boldsymbol{x}$. 
Consequently, only Pauli strings $P_{\boldsymbol{x},\boldsymbol{z}}$ with the same fixed $\boldsymbol{x}$ can appear in its expansion.

To identify these coefficients, use the stabilizer expansion of the projector,
\begin{equation}
|C_\otimes\rangle\!\langle C_\otimes|=\frac{1}{2^n}\sum_{\boldsymbol{x}\in\mathbb Z_2^n}s_{\boldsymbol{x}}^\otimes=\frac{1}{2^n}\sum_{\boldsymbol{x},\boldsymbol{z}\in\mathbb Z_2^n}\alpha_{\boldsymbol{x}}^\otimes(\boldsymbol{z})P_{\boldsymbol{x},\boldsymbol{z}}.
\end{equation}
Meanwhile, its Pauli expansion is $|C_\otimes\rangle\!\langle C_\otimes|=2^{-n}\sum_{\boldsymbol{x},\boldsymbol{z}\in\mathbb Z_2^n}\langle C_\otimes|P_{\boldsymbol{x},\boldsymbol{z}}|C_\otimes\rangle P_{\boldsymbol{x},\boldsymbol{z}}$. Uniqueness of the Pauli expansion gives the stated expression for $\alpha_{\boldsymbol{x}}^\otimes(\boldsymbol{z})$, and the same argument gives the expression for $\alpha_{\boldsymbol{x}}(\boldsymbol{z})$. Applying $\mathcal D^\dagger$ term by term proves Eq.~\eqref{eq:basic-noise-sum}.
\end{proof}

\subsubsection{Statistics for the correct target: $\langle s_{\boldsymbol{x}}\rangle_{\sigma_l}$}

\begin{theorem}
\label{thm:moment-local-product}
The moment-generating polynomial can be expressed as a matrix product
\begin{equation}
\begin{aligned}
G(\xi)&\coloneqq\sum_{k=0}^{n}L_k\langle S^{(k)}\rangle_{\sigma_l}\xi^k = \frac{1}{4^n}\sum_{\boldsymbol{x},\boldsymbol{z},\boldsymbol{u},\boldsymbol{v}}\xi^{|\boldsymbol{x}|}\lambda^{|\boldsymbol{x}\vee\boldsymbol{z}|}(-1)^{\boldsymbol{z}\cdot(\boldsymbol{u}+\boldsymbol{v})+f(\boldsymbol{u})+f(\boldsymbol{u}+\boldsymbol{x})+f(\boldsymbol{v})+f(\boldsymbol{v}+\boldsymbol{x})} =\underbrace{\begin{bmatrix}\cdots&w_\xi(q_1)w_\xi(q_2)&\cdots\end{bmatrix}}_{\boldsymbol{p}_\xi^{\mathsf T}\;(1\times256)}\underbrace{T(\xi)^{n-2}}_{256\times256}
\underbrace{\begin{bmatrix}1\\\vdots\\1\end{bmatrix}}_{\boldsymbol{1}\;(256\times1)},
\end{aligned}
\label{eq:moment-transfer-matrix}
\end{equation}
where $q_i\coloneqq(x_i,z_i,u_i,v_i)\in\mathbb Z_2^4$ has $16$ possible configurations, so the pair $(q_i,q_{i+1})$ indexing a row or column has $256$ possible configurations. For each configuration, define the local weight $w_\xi(q_i)\coloneqq4^{-1}\xi^{x_i}\lambda^{x_i\vee z_i}(-1)^{z_i(u_i+v_i)}$. 
Then, $\boldsymbol{p}_\xi$ and $T(\xi)$ can be explicitly constructed. The nonzero entries are $\left[T(\xi)\right]_{(q_i,q_{i+1}),(q_{i+1},q_{i+2})}=w_\xi(q_{i+2})(-1)^{f_3(q_i,q_{i+1},q_{i+2})}$, where $f_3(q_i,q_{i+1},q_{i+2})\coloneqq u_iu_{i+1}u_{i+2}+(u_i+x_i)(u_{i+1}+x_{i+1})(u_{i+2}+x_{i+2})+v_iv_{i+1}v_{i+2}+(v_i+x_i)(v_{i+1}+x_{i+1})(v_{i+2}+x_{i+2})$. 
Hence, $G(\xi)$ can be evaluated by matrix product, rather than by explicitly summing over $2^{4n}$ terms. For $k=0,\ldots,n$, the stabilizer moments are given by the coefficients of $G(\xi)$ as
\begin{equation}
\langle S^{(k)}\rangle_{\sigma_l}
=\frac{[\xi^k]G(\xi)}{L_k}
=\frac{1}{(n+1)L_k}\sum_{j=0}^{n}G(\xi_j)
\exp\!\left(-\frac{2\pi\mathrm{i}jk}{n+1}\right),
\end{equation}
where $\xi_j\coloneqq\exp\!\left(2\pi\mathrm{i}j/(n+1)\right)$. All $\langle S^{(k)}\rangle_{\sigma_l}$ can be computed using Algorithm~\ref{alg:moment-transfer}. 
\end{theorem}

\begin{proof}
We first express $G(\xi)$ as a product of local factors and contract them using $T(\xi)$. We then recover its coefficients by a discrete Fourier transform. 
Firstly, the identity $P_{\boldsymbol{x},\boldsymbol{z}}|\boldsymbol{u}\rangle=i^{\boldsymbol{x}\cdot\boldsymbol{z}}(-1)^{\boldsymbol{z}\cdot\boldsymbol{u}}|\boldsymbol{u}+\boldsymbol{x}\rangle$ and Eq.~\eqref{eq:rank-adaptive-C3} give
\begin{equation}
\alpha_{\boldsymbol{x}}(\boldsymbol{z})=\langle C|P_{\boldsymbol{x},\boldsymbol{z}}|C\rangle = \frac{i^{\boldsymbol{x}\cdot\boldsymbol{z}}}{2^n}\sum_{\boldsymbol{u}\in\mathbb Z_2^n}(-1)^{f(\boldsymbol{u})+f(\boldsymbol{u}+\boldsymbol{x})+\boldsymbol{z}\cdot\boldsymbol{u}}.
\label{eq:general-pauli-coefficient}
\end{equation}
By the definition of $S^{(k)}$ and Lemma~\ref{lem:S-pauli-expansion}, we have
\begin{equation}
G(\xi)=\sum_{k=0}^{n}L_k\langle S^{(k)}\rangle_{\sigma_l}\xi^k=\sum_{\boldsymbol{x}\in\mathbb Z_2^n}\langle s_{\boldsymbol{x}}\rangle_{\sigma_l}\xi^{|\boldsymbol{x}|}=\sum_{\boldsymbol{x},\boldsymbol{z}\in\mathbb Z_2^n}\lambda^{|\boldsymbol{x}\vee\boldsymbol{z}|}\left|\alpha_{\boldsymbol{x}}(\boldsymbol{z})\right|^2\xi^{|\boldsymbol{x}|} =\frac{1}{4^n}\sum_{\boldsymbol{x},\boldsymbol{z},\boldsymbol{u},\boldsymbol{v}}\xi^{|\boldsymbol{x}|}\lambda^{|\boldsymbol{x}\vee\boldsymbol{z}|}(-1)^{f(\boldsymbol{u})+f(\boldsymbol{u}+\boldsymbol{x})+f(\boldsymbol{v})+f(\boldsymbol{v}+\boldsymbol{x}) + \boldsymbol{z}\cdot(\boldsymbol{u}+\boldsymbol{v})}.
\label{D11}
\end{equation}
By the definition of $f_3$,
\begin{equation}
\begin{aligned}
f(\boldsymbol{u})+f(\boldsymbol{u}+\boldsymbol{x})+f(\boldsymbol{v})+f(\boldsymbol{v}+\boldsymbol{x})\equiv\sum_{i=1}^{n-2}f_3(q_i,q_{i+1},q_{i+2})\pmod 2.
\end{aligned}
\end{equation}
Using this local decomposition, the sums in Eq.~\eqref{D11} can be arranged from left to right as
\begin{equation}
G(\xi)=\sum_{q_1,q_2\in\mathbb Z_2^4}w_\xi(q_1)w_\xi(q_2)\sum_{q_3\in\mathbb Z_2^4}w_\xi(q_3)(-1)^{f_3(q_1,q_2,q_3)}\cdots\sum_{q_n\in\mathbb Z_2^4}w_\xi(q_n)(-1)^{f_3(q_{n-2},q_{n-1},q_n)},
\label{eq:moment-nested-sums}
\end{equation}
where $w_\xi(q_i)\coloneqq4^{-1}\xi^{x_i}\lambda^{x_i\vee z_i}(-1)^{z_i(u_i+v_i)}$.
The first factor defines the boundary row vector $\boldsymbol{p}_\xi^{\mathsf T}$ with entries indexed by $(q_1,q_2)$. Applying these $n-2$ transfers and then summing the entries of the resulting row vector yields Eq.~\eqref{eq:moment-transfer-matrix}.

Since $G(\xi)$ can be efficiently calculated in a matrix-product form, we now recover the coefficients $[\xi^k]G(\xi)=L_k\langle S^{(k)}\rangle_{\sigma_l}$ from its values at distinct points.
Let $\xi_j\coloneqq\exp\!\left(\frac{2\pi \mathrm{i}j}{n+1}\right)$ for $j=0,\ldots,n$. The transfer-matrix contractions at these points give the discrete Fourier transform
\begin{equation}
G(\xi_j)=\sum_{k=0}^{n}L_k\langle S^{(k)}\rangle_{\sigma_l}\exp\!\left(\frac{2\pi \mathrm{i}jk}{n+1}\right),\qquad L_k\langle S^{(k)}\rangle_{\sigma_l}=\frac{1}{n+1}\sum_{j=0}^{n}G(\xi_j)\exp\!\left(-\frac{2\pi \mathrm{i}jk}{n+1}\right).
\label{eq:moment-fourier-reconstruction}
\end{equation}
In exact arithmetic, a single FFT recovers all $\langle S^{(k)}\rangle_{\sigma_l}$. 
\end{proof}

\begin{center}
\begin{minipage}{0.96\linewidth}
\refstepcounter{algorithm}\label{alg:moment-transfer}
\hrule height 0.8pt
\vspace{0.2em}
\noindent\textbf{Algorithm \thealgorithm}\quad Simulation of stabilizer moments under local depolarizing noise
\vspace{0.2em}
\hrule height 0.4pt
\vspace{0.2em}
\footnotesize
\begin{algorithmic}[1]
\REQUIRE Number of qubits $n\geq3$ and $\lambda = 1-4p/3$
\ENSURE Expectation values $\{\langle S^{(k)}\rangle_{\sigma_l}\}_{k=0}^{n}$ 
\STATE Construct the sparse transfer matrix in Eq.~\eqref{eq:moment-transfer-matrix}
\FOR{$j=0,\ldots,n$}
    \STATE $\xi_j\gets\exp[2\pi\mathrm{i}j/(n+1)]$, and $G(\xi_j)\gets\boldsymbol{p}_{\xi_j}^{\mathsf T}T(\xi_j)^{n-2}\boldsymbol{1}$
\ENDFOR
\STATE Compute $\langle S^{(k)}\rangle_{\sigma_l}  \gets[(n+1)L_k]^{-1}\sum_{j=0}^{n}G(\xi_j)\exp[-2\pi\mathrm{i}jk/(n+1)]$ for all $k=0,\ldots,n$ using one FFT
\STATE \textbf{return} $\{\langle S^{(k)}\rangle_{\sigma_l}\}_{k=0}^{n}$ 
\end{algorithmic}
\vspace{0.25em}
\hrule height 0.8pt
\end{minipage}
\end{center}

For large $n$, FFT reconstruction may lose numerical accuracy due to cancellation. One can instead retain the coefficients of each power of $\xi$ throughout the transfer-matrix contractions and divide the final coefficients by $L_k$ to obtain the moments directly.

\subsubsection{Statistics for the wrong target: $\langle s_{\boldsymbol{x}}^\otimes\rangle_{\sigma_l}$}

We take $n=100$ as an example. The prepared state remains $\sigma_l$ defined in Eq.~\eqref{eq:appD3}, obtained by applying local depolarizing noise to the connected chain $|C\rangle$, whereas the candidate target is the disconnected chain $|C_\otimes\rangle=\bigl(|C^{(25)}\rangle\bigr)^{\otimes4}$, i.e.,
\begin{equation}
\begin{aligned}
|C_\otimes\rangle=V|C\rangle=2^{-50}\sum_{\boldsymbol{y}\in\mathbb Z_2^{100}}(-1)^{f_\otimes(\boldsymbol{y})}|\boldsymbol{y}\rangle,\qquad
V\coloneqq\prod_{c\in\{25,50,75\}}C_{\{c-1,c,c+1\}}C_{\{c,c+1,c+2\}}.
\end{aligned}
\end{equation}
where its phase function is $f_\otimes(\boldsymbol{y})\coloneqq f(\boldsymbol{y})-\sum_{i\in\{24,25,49,50,74,75\}}y_iy_{i+1}y_{i+2}\pmod 2$. 
Specializing Eq.~\eqref{eq:general-pauli-coefficient} to $f_\otimes$, we obtain $\alpha_{\boldsymbol{x}}^\otimes(\boldsymbol{z})=2^{-100}i^{\boldsymbol{x}\cdot\boldsymbol{z}}\sum_{\boldsymbol{u}\in\mathbb Z_2^{100}}(-1)^{f_\otimes(\boldsymbol{u})+f_\otimes(\boldsymbol{u}+\boldsymbol{x})+\boldsymbol{z}\cdot\boldsymbol{u}}$. Applying Lemma~\ref{lem:S-pauli-expansion} then yields
\begin{equation}
\begin{aligned}
G_\otimes(\xi)&\coloneqq\sum_{k=0}^{100}L_k\langle S_\otimes^{(k)}\rangle_{\sigma_l}\xi^k=\sum_{\boldsymbol{x},\boldsymbol{z}\in\mathbb Z_2^{100}}\lambda^{|\boldsymbol{x}\vee\boldsymbol{z}|}\alpha_{\boldsymbol{x}}^\otimes(\boldsymbol{z})\alpha_{\boldsymbol{x}}(\boldsymbol{z})\xi^{|\boldsymbol{x}|} =\frac{1}{4^{100}}\sum_{\boldsymbol{x},\boldsymbol{z},\boldsymbol{u},\boldsymbol{v}\in\mathbb Z_2^{100}}\xi^{|\boldsymbol{x}|}\lambda^{|\boldsymbol{x}\vee\boldsymbol{z}|}(-1)^{\boldsymbol{z}\cdot(\boldsymbol{u}+\boldsymbol{v})+f_\otimes(\boldsymbol{u})+f_\otimes(\boldsymbol{u}+\boldsymbol{x})+f(\boldsymbol{v})+f(\boldsymbol{v}+\boldsymbol{x})},
\end{aligned}
\label{eq:cross-moment-generating-polynomial}
\end{equation}
Here, both $\alpha_{\boldsymbol{x}}^\otimes(\boldsymbol{z})$ and $\alpha_{\boldsymbol{x}}(\boldsymbol{z})$ are real by definition, so every nonzero contribution has $\boldsymbol{x}\cdot\boldsymbol{z}$ even and the product of the two phase factors $i^{\boldsymbol{x}\cdot\boldsymbol{z}}$ is $+1$.

Similarly, we can factor the summand into local terms as in Eq.~\eqref{eq:moment-nested-sums}
\begin{equation}
G_\otimes(\xi)=\boldsymbol{p}_\xi^{\mathsf T}T(\xi)^{23}T_\otimes(\xi)^2T(\xi)^{23}T_\otimes(\xi)^2T(\xi)^{23}T_\otimes(\xi)^2T(\xi)^{23}\boldsymbol{1},
\label{eq:cross-transfer-contraction}
\end{equation}
where the nonzero entries of $T_\otimes(\xi)$ are
\begin{equation}
\left[T_\otimes(\xi)\right]_{(q_i,q_{i+1}),(q_{i+1},q_{i+2})}\coloneqq w_\xi(q_{i+2})(-1)^{f_3(q_i,q_{i+1},q_{i+2})-\left[u_iu_{i+1}u_{i+2}+(u_i+x_i)(u_{i+1}+x_{i+1})(u_{i+2}+x_{i+2})\right]}.
\end{equation}
The transfer $T_\otimes(\xi)$ occurs at $i=24,25,49,50,74,75$, while $T(\xi)$ occurs at every other position.
Similarly, a single FFT recovers all $L_k\langle S_\otimes^{(k)}\rangle_{\sigma_l}$ from the $101$ values $G_\otimes(\xi_j)$, $j=0,\ldots,100$.
For numerical evaluation, the polynomial coefficients can instead be propagated through the transfer sequence in Eq.~\eqref{eq:cross-transfer-contraction}.

\subsection{Statistics under global depolarizing noise}

For the connected target $|C\rangle$, global depolarizing noise directly gives $\langle S^{(k)}\rangle_{\sigma_g}=1-\eta$ for all $k \geq 1$. 

We now consider the disconnected candidate $|C_\otimes\rangle$. Under global depolarizing noise, we have $\langle S_\otimes^{(k)}\rangle_{\sigma_g}=(1-\eta)\langle C|S_\otimes^{(k)}|C\rangle$ for $k\geq 1$. 
We first evaluate $\langle C|s_{\boldsymbol{x}}^\otimes|C\rangle$. Since $s_{\boldsymbol{x}}^\otimes=U_\otimes X^{\boldsymbol{x}}U_\otimes^\dagger$, we obtain
\begin{equation}
\langle C|s_{\boldsymbol{x}}^\otimes|C\rangle = \langle C_\otimes | V^\dagger s_{\boldsymbol{x}}^\otimes V|C_\otimes \rangle = \langle + |^{\otimes 100} U_\otimes^\dagger V^\dagger s_{\boldsymbol{x}}^\otimes V U_\otimes| + \rangle^{\otimes 100} = \langle+|^{\otimes100}V X^{\boldsymbol{x}}V|+\rangle^{\otimes100}.
\label{eq:global-cross-stabilizer-overlap}
\end{equation}
Recall that $V = \prod_{c\in\{25,50,75\}}C_{\{c-1,c,c+1\}}C_{\{c,c+1,c+2\}}$. 
For each cut, relabel its four sites $\{c-1,c,c+1,c+2\}$ as $1,2,3,4$ and write $V_c=C_{\{1,2,3\}}C_{\{2,3,4\}}$. Let $\boldsymbol{x}_c\in\mathbb Z_2^4$ denote the stabilizer label on these sites. Since the three cut neighborhoods are disjoint, the remaining $(100 - 3\times 4)$ qubits do not affect Eq.~\eqref{eq:global-cross-stabilizer-overlap}. For fixed $\boldsymbol{x}_{25},\boldsymbol{x}_{50},\boldsymbol{x}_{75}$, there are $\binom{88}{k-\sum_{c\in\{25,50,75\}}|\boldsymbol{x}_c|}$ assignments of $\boldsymbol{x}$ of weight $k$. 
Hence,
\begin{equation}
\begin{aligned}
\langle S_\otimes^{(k)}\rangle_{\sigma_g}
&=\frac{1-\eta}{\binom{100}{k}}
\sum_{(\boldsymbol{x}_{25},\boldsymbol{x}_{50},\boldsymbol{x}_{75})\in(\mathbb Z_2^4)^3}
\binom{88}{k-\sum_{c\in\{25,50,75\}}|\boldsymbol{x}_c|}
\prod_{c\in\{25,50,75\}}\langle+|^{\otimes4}V_cX^{\boldsymbol{x}_c}V_c|+\rangle^{\otimes4},\qquad k\geq1,
\end{aligned}
\label{eq:disconnected-target-global-statistics}
\end{equation}
where $\binom{88}{r}=0$ for $r\notin\{0,\ldots,88\}$. Here, $\langle+|^{\otimes4}V_cX^{\boldsymbol{x}_c}V_c|+\rangle^{\otimes4} $ equals $1$ for $\boldsymbol{x}_c\in\{0000,1001\}$ and $1/2$ otherwise.

\subsection{Direct fidelity estimation for the 1D order-3 cluster state}
\label{app:dfe-1d-order3}

In the standard DFE~\cite{prl_2011_directfidelity}, the characteristic function is given by $\chi_{|C\rangle\langle C|}(\boldsymbol{x},\boldsymbol{z}) \coloneqq \langle C | P_{\boldsymbol{x},\boldsymbol{z}} |C\rangle /\sqrt{2^n} = \alpha_{\boldsymbol{x}}(\boldsymbol{z})/\sqrt{2^n}$. Let $\epsilon_{\rm DFE}$ and $\delta_{\rm DFE}$ denote the accuracy and failure parameters, for $\ell=\lceil1/(\epsilon_{\rm DFE}^2\delta_{\rm DFE})\rceil$ sampled Pauli observables $P_{\boldsymbol{x},\boldsymbol{z}}$, each sampled observable is estimated using $m_{\boldsymbol{x},\boldsymbol{z}} =\left\lceil\frac{2\ln(2/\delta_{\rm DFE})}{\ell\epsilon_{\rm DFE}^2|\alpha_{\boldsymbol{x}}(\boldsymbol{z})|^2}\right\rceil$ copies. After that, writing the resulting estimator as $\widetilde Y$, the fidelity lies in the range $[\widetilde Y-2\epsilon_{\rm DFE},\widetilde Y+2\epsilon_{\rm DFE}]$ with probability $\geq 1-2\delta_{\rm DFE}$. We set $\epsilon_{\rm DFE}=\frac{\upDelta}{2}$ and $\delta_{\rm DFE}=\frac{\delta}{2}$. Then, the resulting interval $[\widetilde Y-\upDelta,\widetilde Y+\upDelta]$ satisfies $F_{\max}-F_{\min}=2\upDelta$ with probability at least $1-\delta$. 
To achieve this accuracy, the DFE protocol requires, on average,
\begin{equation}
\mathbb{E}[N_{\rm DFE}]
=\ell \sum_{\boldsymbol{x},\boldsymbol{z}} [\chi_{|C\rangle\langle C|}(\boldsymbol{x},\boldsymbol{z})]^2 m_{\boldsymbol{x},\boldsymbol{z}}
=\ell\sum_{\alpha_{\boldsymbol{x}}(\boldsymbol{z})\neq0}
\frac{|\alpha_{\boldsymbol{x}}(\boldsymbol{z})|^2}{2^n}
m_{\boldsymbol{x},\boldsymbol{z}}\geq
\frac{2\widetilde R_n}{\epsilon_{\rm DFE}^2}
\ln\frac{2}{\delta_{\rm DFE}}
=\frac{8\widetilde R_n}{\upDelta^2}\ln\frac{4}{\delta}, \qquad \widetilde R_n\coloneqq R_n/2^n. 
\end{equation}
Here, $R_n\coloneqq\sum_{\boldsymbol{x}\in\mathbb Z_2^n}\left|\left\{\boldsymbol{z}\in\mathbb Z_2^n:\alpha_{\boldsymbol{x}}(\boldsymbol{z})\neq0\right\}\right|$ denotes the total number of nonzero Pauli coefficients.

In the following, we show how to efficiently compute $\widetilde R_n$ for the 1D order-3 cluster state when $n$ is large. 

\begin{lemma}
\label{lem:dfe-pauli-coefficient}
For $\boldsymbol{x},\boldsymbol{z}\in\mathbb Z_2^n$, define $d(\boldsymbol{x})\in\mathbb Z_2^n$ and the matrices $B(\boldsymbol{x}),C(\boldsymbol{x})\in\mathbb Z_2^{n\times n}$ by
\begin{equation}
d_i(\boldsymbol{x})\coloneqq x_{i-2}x_{i-1}+x_{i-1}x_{i+1}+x_{i+1}x_{i+2},\qquad B_{ij}(\boldsymbol{x})\coloneqq\begin{cases}\displaystyle\sum_{k:\,\{i,j,k\}\in E_3}x_k,&i<j,\\0,&i\geq j,\end{cases}\qquad C(\boldsymbol{x})\coloneqq B(\boldsymbol{x})+B(\boldsymbol{x})^{\mathsf T}.
\label{eq:dfe-quadratic-matrices}
\end{equation}
Here, $x_j=0$ for $j\notin[n]$, and all arithmetic in $d(\boldsymbol{x})$, $B(\boldsymbol{x})$, and $C(\boldsymbol{x})$ is over $\mathbb Z_2$.
Then, we have
\begin{equation}
\alpha_{\boldsymbol{x}}(\boldsymbol{z})=\langle C|P_{\boldsymbol{x},\boldsymbol{z}}|C\rangle
=\frac{i^{\boldsymbol{x}\cdot\boldsymbol{z}}}{2^n}
\sum_{\boldsymbol{u}\in\mathbb Z_2^n}
(-1)^{\boldsymbol{u}^{\mathsf T}B(\boldsymbol{x})\boldsymbol{u}+[d(\boldsymbol{x})+\boldsymbol{z}]\cdot\boldsymbol{u} + f(\boldsymbol{x})}.
\label{eq:dfe-pauli-coefficient}
\end{equation}
\end{lemma}

\begin{proof}
The lemma follows by substituting the identity
\begin{equation}
\begin{aligned}
f(\boldsymbol{u})+f(\boldsymbol{u}+\boldsymbol{x})
&\equiv\sum_{i=1}^{n-2}\bigl[(x_i u_{i+1}u_{i+2}+x_{i+1}u_i u_{i+2}+x_{i+2}u_i u_{i+1}) + (u_ix_{i+1}x_{i+2}+x_iu_{i+1}x_{i+2}+x_ix_{i+1}u_{i+2}) + x_ix_{i+1}x_{i+2}\bigr]\\
&=\boldsymbol{u}^{\mathsf T}B(\boldsymbol{x})\boldsymbol{u}+d(\boldsymbol{x})\cdot\boldsymbol{u}+f(\boldsymbol{x})\pmod 2.
\end{aligned}
\end{equation}
into Eq.~\eqref{eq:general-pauli-coefficient}.
\end{proof}

\begin{lemma}
\label{thm:dfe-rank-sum}
Let $r(\boldsymbol{x})\coloneqq\operatorname{rank}_{\mathbb Z_2}C(\boldsymbol{x})$. Then
\begin{equation}
R_n\coloneqq\sum_{\boldsymbol{x}\in\mathbb Z_2^n}\left|\left\{\boldsymbol{z}\in\mathbb Z_2^n:\alpha_{\boldsymbol{x}}(\boldsymbol{z})\neq0\right\}\right|=\sum_{\boldsymbol{x}\in\mathbb Z_2^n}2^{r(\boldsymbol{x})}.
\label{eq:Rn-rank-sum}
\end{equation}
\end{lemma}

\begin{proof}
We first characterize the support of the Pauli coefficients for a fixed $\boldsymbol{x}$. Equation~\eqref{eq:dfe-pauli-coefficient} gives
\begin{equation}
\begin{aligned}
|\alpha_{\boldsymbol{x}}(\boldsymbol{z})|^2 \propto \sum_{\boldsymbol{u},\boldsymbol{u}'\in\mathbb Z_2^n}(-1)^{\boldsymbol{u}^{\mathsf T}B(\boldsymbol{x})\boldsymbol{u}+\boldsymbol{u}'^{\mathsf T}B(\boldsymbol{x})\boldsymbol{u}'+[d(\boldsymbol{x})+\boldsymbol{z}]\cdot(\boldsymbol{u}+\boldsymbol{u}')+2f(\boldsymbol{x})} \underset{\boldsymbol{u}'=\boldsymbol{u}+\boldsymbol{y}}{=} \sum_{\boldsymbol{u},\boldsymbol{y}\in\mathbb Z_2^n}
(-1)^{\boldsymbol{u}^{\mathsf T}C(\boldsymbol{x})\boldsymbol{y}
+\boldsymbol{y}^{\mathsf T}B(\boldsymbol{x})\boldsymbol{y}
+[d(\boldsymbol{x})+\boldsymbol{z}]\cdot\boldsymbol{y}} = 2^n\sum_{\boldsymbol{y}\in\ker C(\boldsymbol{x})}
(-1)^{\boldsymbol{y}^{\mathsf T}B(\boldsymbol{x})\boldsymbol{y}
+[d(\boldsymbol{x})+\boldsymbol{z}]\cdot\boldsymbol{y}}.
\end{aligned}
\label{eq:dfe-walsh-square}
\end{equation}
The sum over $\boldsymbol{u}$ is $2^n$ when $C(\boldsymbol{x})\boldsymbol{y}=\boldsymbol{0}$ and vanishes otherwise. We then define the exponent function 
$\varphi_{\boldsymbol{x}}(\boldsymbol{y})\coloneqq\boldsymbol{y}^{\mathsf T}B(\boldsymbol{x})\boldsymbol{y}+d(\boldsymbol{x})\cdot\boldsymbol{y}$. It is linear on $\ker C(\boldsymbol{x})$ in the sense that for any $\boldsymbol{y},\boldsymbol{y}'\in\ker C(\boldsymbol{x})$, we have 
\begin{equation}
\varphi_{\boldsymbol{x}}(\boldsymbol{y}+\boldsymbol{y}')
=\varphi_{\boldsymbol{x}}(\boldsymbol{y})
+\varphi_{\boldsymbol{x}}(\boldsymbol{y}')
+\boldsymbol{y}^{\mathsf T}C(\boldsymbol{x})\boldsymbol{y}'
=\varphi_{\boldsymbol{x}}(\boldsymbol{y})
+\varphi_{\boldsymbol{x}}(\boldsymbol{y}').
\end{equation}
Hence $\varphi_{\boldsymbol{x}}(\boldsymbol{y})+\boldsymbol{z}\cdot\boldsymbol{y}$ is also linear on the kernel. 
If there exists $\boldsymbol{y}_0\in\ker C(\boldsymbol{x})$ such that $\boldsymbol{z}\cdot\boldsymbol{y}_0\neq\varphi_{\boldsymbol{x}}(\boldsymbol{y}_0)$, then $\varphi_{\boldsymbol{x}}(\boldsymbol{y}_0)+\boldsymbol{z}\cdot\boldsymbol{y}_0=1$ in $\mathbb Z_2$, so
\begin{equation}
|\alpha_{\boldsymbol{x}}(\boldsymbol{z})|^2 \propto  \sum_{\boldsymbol{y}\in\ker C(\boldsymbol{x})}(-1)^{\varphi_{\boldsymbol{x}}(\boldsymbol{y})+\boldsymbol{z}\cdot\boldsymbol{y}}=\sum_{\boldsymbol{y}\in\ker C(\boldsymbol{x})}(-1)^{\varphi_{\boldsymbol{x}}(\boldsymbol{y}+\boldsymbol{y}_0)+\boldsymbol{z}\cdot(\boldsymbol{y}+\boldsymbol{y}_0)}=-\sum_{\boldsymbol{y}\in\ker C(\boldsymbol{x})}(-1)^{\varphi_{\boldsymbol{x}}(\boldsymbol{y})+\boldsymbol{z}\cdot\boldsymbol{y}}=0.
\end{equation}
If no such $\boldsymbol{y}_0$ exists, every summand $(-1)^{\varphi_{\boldsymbol{x}}(\boldsymbol{y})+\boldsymbol{z}\cdot\boldsymbol{y}}$ equals one. Therefore, Eq.~\eqref{eq:dfe-walsh-square} gives
\begin{equation}
\left|\alpha_{\boldsymbol{x}}(\boldsymbol{z})\right|^2\propto\sum_{\boldsymbol{y}\in\ker C(\boldsymbol{x})}
(-1)^{\varphi_{\boldsymbol{x}}(\boldsymbol{y})+\boldsymbol{z}\cdot\boldsymbol{y}}
=\begin{cases}
|\ker C(\boldsymbol{x})|,
&\boldsymbol{z}\cdot\boldsymbol{y}=\varphi_{\boldsymbol{x}}(\boldsymbol{y})
\quad\text{for all }\boldsymbol{y}\in\ker C(\boldsymbol{x}),\\
0,&\text{otherwise}.
\end{cases}
\label{eq:dfe-kernel-sum}
\end{equation}
Therefore, $\alpha_{\boldsymbol{x}}(\boldsymbol{z})$ is nonzero exactly when $\boldsymbol{z}\cdot\boldsymbol{y}=\varphi_{\boldsymbol{x}}(\boldsymbol{y})$ for every $\boldsymbol{y}\in\ker C(\boldsymbol{x})$. 
For fixed $\boldsymbol{x}$, it remains to count the values of $\boldsymbol{z}$ satisfying this condition. 
Choosing a basis $\boldsymbol{y}_1,\ldots,\boldsymbol{y}_{n-r(\boldsymbol{x})}$ of the kernel reduces this condition to the independent equations
\begin{equation}
\boldsymbol{z}\cdot\boldsymbol{y}_j
=\varphi_{\boldsymbol{x}}(\boldsymbol{y}_j),
\qquad j=1,\ldots,n-r(\boldsymbol{x}).
\end{equation}
The row vectors $\boldsymbol{y}_j^{\mathsf T}$ are linearly independent, so the coefficient matrix has full row rank and the system is consistent for any right-hand side. Its solution set is therefore an affine subspace of dimension $r(\boldsymbol{x})$, giving exactly $2^{r(\boldsymbol{x})}$ admissible values of $\boldsymbol{z}$. Summing over $\boldsymbol{x}$ proves Eq.~\eqref{eq:Rn-rank-sum}.
\end{proof}

\begin{theorem}
\label{thm:dfe-rank-transfer}
For every $n\geq4$, the rank sum $R_n = \sum_{\boldsymbol{x}\in\mathbb{Z}_2^n} 2^{r(\boldsymbol{x})}$ is given by a matrix product
\begin{equation}
R_n=\boldsymbol p^{\mathsf T}T^{n-4}\boldsymbol q.
\label{eq:dfe-rank-transfer}
\end{equation}
Here, $\boldsymbol p\in\mathbb R^{1072}$, $T\in\mathbb R^{1072\times1072}$, and $\boldsymbol q\in\mathbb R^{1072}$ are independent of $n$. Their entries are explicitly specified in Eqs.~\eqref{eq:dfe-initial-boundary-vector}, \eqref{eq:dfe-rank-transition-matrix}, and~\eqref{eq:dfe-final-boundary-vector}, respectively. Therefore, the sample complexity of DFE can be evaluated efficiently using algorithm~\ref{alg:dfe-rank-transfer}. 
\end{theorem}

\begin{proof}
First, Eq.~\eqref{eq:dfe-quadratic-matrices} implies that $C(\boldsymbol{x})$ is a symmetric band matrix with zero diagonal. Its pattern and possibly nonzero entries are
\begin{equation}
C(\boldsymbol{x})=
\begin{pmatrix}
0 & * & * &   &   &   &   \\
* & 0 & * & * &   &   &   \\
* & * & 0 & * & * &   &   \\
  & * & * & 0 & * & * &   \\
  &   & \ddots & \ddots & \ddots & \ddots & \ddots \\
  &   &   & * & * & 0 & * \\
  &   &   &   & * & * & 0
\end{pmatrix},
\qquad
\begin{gathered}
C_{j,j+1}(\boldsymbol{x})=C_{j+1,j}(\boldsymbol{x})=x_{j-1}+x_{j+2},\quad 1\leq j\leq n-1,\\
C_{j,j+2}(\boldsymbol{x})=C_{j+2,j}(\boldsymbol{x})=x_{j+1},\quad 1\leq j\leq n-2.
\end{gathered}
\label{eq:C-matrix-order3}
\end{equation}
Here, blank entries are zero and each asterisk denotes the corresponding $\boldsymbol{x}$-dependent bit specified on the right. 
To calculate its rank, Gaussian elimination over $\mathbb Z_2$ begins with the first two rows: $C_{1,*}=(0,x_3,x_2,0,\ldots)$ and $C_{2,*}=(x_3,0,x_1+x_4,x_3,0,\ldots)$. 
Since both rows have no support beyond column $4$, their span is represented by the boundary subspace $\mathcal U_2=\operatorname{span}\!\left\{(0,x_3,x_2,0),(x_3,0,x_1+x_4,x_3)\right\} \subseteq \mathbb{Z}_2^4$. 
Its dimension is exactly the rank of the first two rows. Hence, the element of the initial vector $\boldsymbol{p}$ is given by
\begin{equation}
p_s=
\begin{cases}
2^{\dim\mathcal U_2},
&s=(x_1,x_2,x_3,x_4,\mathcal U_2),\quad
(x_1,x_2,x_3,x_4)\in\mathbb Z_2^4,\\
0,&\text{otherwise}.
\end{cases}
\label{eq:dfe-initial-boundary-vector}
\end{equation}
For $d=0,\ldots,4$, the number of $d$-dimensional subspaces $\mathcal{U}_2$ of $\mathbb Z_2^4$ is the Gaussian binomial coefficient $\genfrac{[}{]}{0pt}{}{4}{d}_2=\prod_{j=0}^{d-1}(2^{4}-2^j)/(2^{d}-2^j)$, which is $1$ for $d = 0$. 
Hence, $\mathbb Z_2^4$ has $\sum_{d=0}^4\genfrac{[}{]}{0pt}{}{4}{d}_2=1+15+35+15+1=67$ linear subspaces in total, and the vector $\boldsymbol{p}$ contains $2^4\times67=1072$ elements.

We next propagate this boundary data from left to right. Suppose that rows $1,\ldots,i$ have been processed, with $i\geq2$. Since no subsequent row has support in columns $1,\ldots,i-2$, only combinations of the processed rows that vanish in these columns can be used to eliminate row $i+1$. Hence, we define
\begin{equation}
\mathcal U_i\coloneqq\left\{(u_{i-1},u_i,u_{i+1},u_{i+2}):
\boldsymbol{u}\in\operatorname{span}\{C_{1,*},\ldots,C_{i,*}\},
u_1=\cdots=u_{i-2}=0\right\}\subseteq\mathbb Z_2^4.
\label{eq:dfe-boundary-subspace}
\end{equation}
The next Gaussian elimination step uses only the entries shown below.
\begin{equation}
\begin{array}{c|c|cccc|c|c}
\text{column index}
&1{:}i-2&i-1&i&i+1&i+2&i+3&i+4{:}n\\ \hline
\boldsymbol u\in\mathcal U_i
&\boldsymbol 0&u_{i-1}&u_i&u_{i+1}&u_{i+2}&0&\boldsymbol 0\\
C_{i+1,*}
&\boldsymbol 0&x_i&x_{i-1}+x_{i+2}&0&x_i+b&x_{i+2}&\boldsymbol 0
\end{array}
\label{eq:dfe-boundary-membership-test}
\end{equation}
At this step, $b=x_{i+3}$ is the only newly appended bit. 
Let $\boldsymbol c_i(b) = (x_i, x_{i-1}+x_{i+2}, 0, x_i+b, x_{i+2})$ denote the five possibly nonzero entries of $C_{i+1,*}$ in columns $i-1,\ldots,i+3$. 
The new row leaves the rank unchanged if $\boldsymbol c_i(b)$ can be expressed as a linear combination of vectors in $\mathcal U_i\times\{0\}$, and increases it by one otherwise. Accordingly, for $s=(x_{i-1},x_i,x_{i+1},x_{i+2},\mathcal U_i)$, the rank increment is $\upDelta r(s,b)\coloneqq\mathds{1}[\boldsymbol c_i(b)\notin\mathcal U_i\times\{0\}]$. 
To move the active window one site to the right, we discard the column $i - 1$, and the resulting boundary subspace is
\begin{equation}
\mathcal U_{i+1}=\left\{(v_2,v_3,v_4,v_5):\boldsymbol{v}\in\operatorname{span}\!\left[(\mathcal U_i\times\{0\})\cup\{\boldsymbol{c}_i(b)\}\right],\ v_1=0\right\}\subseteq\mathbb Z_2^4,
\label{eq:dfe-boundary-update}
\end{equation}
so the updated state is $s \rightarrow s' = \Phi_b(s)\coloneqq(x_i,x_{i+1},x_{i+2},b,\mathcal U_{i+1})$. The corresponding transition matrix is
\begin{equation}
T_{s,s'}\coloneqq
\sum_{\substack{b\in\mathbb Z_2:\,
\Phi_b(s)=s'}}
2^{\upDelta r(s,b)}.
\label{eq:dfe-rank-transition-matrix}
\end{equation}
Each row of $T$ contains two transitions, corresponding to $b=0$ and $b=1$. The matrix $T$ can be constructed by direct enumeration. 

At the right boundary, write $s=(x_{n-3},x_{n-2},x_{n-1},x_n,\mathcal U_{n-2})$. The final two rows $C_{n-1,*}$ and $C_{n,*}$, restricted to the possibly nonzero columns of the matrix, are $\boldsymbol{c}_{\rm R}^{(1)}=(x_{n-2},x_{n-3}+x_n,0,x_{n-2})$, and $\boldsymbol{c}_{\rm R}^{(2)}=(0,x_{n-1},x_{n-2},0)$. 
Their rank increments define
\begin{equation}
q_s=2^{\mathds{1}[\boldsymbol{c}_{\rm R}^{(1)}\notin\mathcal U_{n-2}]
+\mathds{1}[\boldsymbol{c}_{\rm R}^{(2)}\notin
\operatorname{span}(\mathcal U_{n-2}\cup\{\boldsymbol{c}_{\rm R}^{(1)}\})]}.
\label{eq:dfe-final-boundary-vector}
\end{equation}
The initial vector $\boldsymbol p$ selects $(x_1,x_2,x_3,x_4)$ and contributes $2^{\dim\mathcal U_2}$, which accounts for the rank of the first two rows. The $n-4$ transitions append $x_5,\ldots,x_n$ and contribute the rank increments of rows $3,\ldots,n-2$, while $\boldsymbol q$ contributes those of the final two rows. Thus, complete paths are in one-to-one correspondence with bit strings $\boldsymbol{x}$, and each path has weight $2^{r(\boldsymbol{x})}$. Expanding the matrix product therefore gives $\boldsymbol p^{\mathsf T}T^{n-4}\boldsymbol q=\sum_{\boldsymbol{x}}2^{r(\boldsymbol{x})}=R_n$, proving Eq.~\eqref{eq:dfe-rank-transfer}.
\end{proof}

\begin{center}
\begin{minipage}{0.96\linewidth}
\refstepcounter{algorithm}\label{alg:dfe-rank-transfer}
\hrule height 0.8pt
\vspace{0.45em}
\noindent\textbf{Algorithm \thealgorithm}\quad Transfer-matrix evaluation of $R_n$
\vspace{0.4em}
\hrule height 0.4pt
\vspace{0.45em}
\small
\begin{algorithmic}[1]
\REQUIRE Number of qubits $n\geq4$
\ENSURE $R_n=\sum_{\boldsymbol{x}}2^{r(\boldsymbol{x})}$
\STATE $\mathsf{UList}\gets$ the list of all $67$ linear subspaces of $\mathbb Z_2^4$, and $\mathsf{StateList}\gets\mathbb Z_2^4\times\mathsf{UList}$ and index its $1072$ elements
\STATE $\boldsymbol{p}\gets\boldsymbol{0}$
\FOR{each $(x_1,x_2,x_3,x_4)\in\mathbb Z_2^4$}
    \STATE $\mathcal U_2\gets\operatorname{span}\{(0,x_3,x_2,0),(x_3,0,x_1+x_4,x_3)\}$
    \STATE $s\gets(x_1,x_2,x_3,x_4,\mathcal U_2)$ and $p_s\gets p_s+2^{\dim\mathcal U_2}$
\ENDFOR
\STATE $T\gets0_{1072\times1072}$
\FOR{each $s=(x_{i-1},x_i,x_{i+1},x_{i+2},\mathcal U_i)\in\mathsf{StateList}$}
    \FOR{each $b\in\mathbb Z_2$}
        \STATE $\boldsymbol{c}_i(b)\gets(x_i,x_{i-1}+x_{i+2},0,x_i+b,x_{i+2})$ over $\mathbb Z_2$, and $\mathcal U_i^{+}\gets\mathcal U_i\times\{0\}$
        \STATE $\upDelta r\gets\mathds{1}[\boldsymbol{c}_i(b)\notin\mathcal U_i^{+}]$, and $\mathcal U_{i+1}\gets\{(v_2,v_3,v_4,v_5):\boldsymbol{v}\in\operatorname{span}(\mathcal U_i^{+}\cup\{\boldsymbol{c}_i(b)\}),\ v_1=0\}$
        \STATE $s'\gets(x_i,x_{i+1},x_{i+2},b,\mathcal U_{i+1})$, and $T_{s,s'}\gets T_{s,s'}+2^{\upDelta r}$
    \ENDFOR
\ENDFOR
\FOR{each $s=(x_{n-3},x_{n-2},x_{n-1},x_n,\mathcal U_{n-2})\in\mathsf{StateList}$}
    \STATE $\boldsymbol{c}_{\rm R}^{(1)}\gets(x_{n-2},x_{n-3}+x_n,0,x_{n-2})$, and $\boldsymbol{c}_{\rm R}^{(2)}\gets(0,x_{n-1},x_{n-2},0)$
    \STATE $q_s\gets2^{\mathds{1}[\boldsymbol{c}_{\rm R}^{(1)}\notin\mathcal U_{n-2}]+\mathds{1}[\boldsymbol{c}_{\rm R}^{(2)}\notin\operatorname{span}(\mathcal U_{n-2}\cup\{\boldsymbol{c}_{\rm R}^{(1)}\})]}$
\ENDFOR
\STATE \textbf{return} $R_n=\boldsymbol{p}^{\mathsf T}T^{n-4}\boldsymbol{q}$
\end{algorithmic}
\vspace{0.25em}
\hrule height 0.8pt
\end{minipage}
\end{center}

\section{\label{app:dual}An analytical bound from $2n+1$ stabilizing operators}
The dual LP associated with Eq.~\eqref{eq:LP1} is
\begin{equation}
    \begin{aligned}
    \label{eq:dual}
        {F}_{\min}^\star = \max_{y_0,\{y_{\boldsymbol{x}}\}}& \quad  y_{0} + \sum_{\boldsymbol{x}\in\mathcal{M}} \langle s_{\boldsymbol{x}}\rangle _\sigma y_{\boldsymbol{x}}\\
        {\rm s.t. } & \quad y_{0} + \sum_{\boldsymbol{x}\in \mathcal{M}} y_{\boldsymbol{x}} (-1)^{\boldsymbol{x}\cdot \boldsymbol{\lambda}} \leq \delta_{\boldsymbol{\lambda}, \boldsymbol{0}},\quad \forall \boldsymbol{\lambda}\in \mathbb{Z}_2^n,
    \end{aligned}
\end{equation}
where $\mathcal{M} \subseteq \mathbb{Z}_2^n$ labels the set of measured stabilizing operators. 

\begin{lemma}
    Let $n\geq3$ and suppose that the set $\{g_i\}_{i=1}^n \cup \{g_1g_2\cdots g_n\} \cup \{\prod_{j\neq i}g_j\}_{i=1}^n$ is measured. Equivalently, $|\mathcal{M}|=2n+1$ and $\mathcal{M}=\{e_1,\ldots,e_n\}\cup\{\boldsymbol{1}\}\cup\{\boldsymbol{1}-e_1,\ldots,\boldsymbol{1}-e_n\}\subset\mathbb{Z}_2^n$. Here, $e_i$ denotes the standard basis vector with a $1$ in the $i$-th position and $0$ elsewhere. Then the primal optimum in Eq.~\eqref{eq:LP1} obeys
    \begin{equation}
    \label{eq:app_A21}
        F_{\min} \geq \frac{4-n}{8}\left(1+\langle g_1\cdots g_n\rangle_\sigma\right)
        +\frac{1}{8}\left(\sum_{i=1}^n\langle g_i\rangle_\sigma
        +\sum_{i=1}^n\left\langle\prod_{j\neq i}g_j\right\rangle_\sigma\right)
        =\frac{4-n}{8}\left(1+\left\langle S^{(n)}\right\rangle_\sigma\right)
        +\frac{1}{8}\left(L_1\left\langle S^{(1)}\right\rangle_\sigma
        +L_{n-1}\left\langle S^{(n-1)}\right\rangle_\sigma\right).
    \end{equation}
\end{lemma}
\begin{proof}
    We now construct a feasible dual point for the dual LP. Set
    \begin{equation}
        \tilde{y}_0 = {4 - n \over 8}, \quad \tilde{y}_{\boldsymbol{x} = e_i} = \tilde{y}_{\boldsymbol{x} = \boldsymbol{1} - e_i} = {1\over 8},\quad \tilde{y}_{\boldsymbol{x} = \boldsymbol{1}} = {4 - n \over 8}, \quad \forall i \in [n],
    \end{equation}
    whose objective value is the right-hand side of Eq.~\eqref{eq:app_A21}. For $\boldsymbol{\lambda}=\boldsymbol{0}$, the dual constraint is saturated:
    \begin{equation}
        \tilde{y}_0 + \sum_{\boldsymbol{x}\in\mathcal{M}} \tilde{y}_{\boldsymbol{x}} (-1)^{\boldsymbol{x}\cdot \boldsymbol{\lambda}} = {4 - n \over 8} + {n \over 8} + {n \over 8} + {4-n\over 8} = 1.
    \end{equation}
For $\boldsymbol{\lambda}$ of Hamming weight $k\geq1$, we have 
\begin{equation}
    \tilde{y}_0 + \sum_{\boldsymbol{x}\in\mathcal{M}} \tilde{y}_{\boldsymbol{x}} (-1)^{\boldsymbol{x}\cdot \boldsymbol{\lambda}} = {4 - n \over 8} + {n - 2k \over 8} + (-1)^k{n - 2k \over 8} + (-1)^{k}{4 - n\over 8} = \begin{cases}
        {2 - k\over 2}, & k\text{ is even}, \\
        {0}, & k \text{ is odd}.\\
    \end{cases} 
\end{equation}
The dual point is therefore feasible, and weak duality proves Eq.~\eqref{eq:app_A21}. 
\end{proof}

\section{Tightness of bounds from stabilizer moments\label{app:tightness}}

We show that, when all terms in $S^{(k)}$ have the same expectation value, replacing their individual constraints by the stabilizer-moment constraint $\langle S^{(k)}\rangle$ leaves both fidelity extrema unchanged.

\begin{lemma}
\label{lem:equiv_individual_moments}
Fix $\mathcal{K}\subseteq[n]$ and suppose that
$\langle s_{\boldsymbol{x}}\rangle_\sigma=c_k$ for every
$\boldsymbol{x}$ with $|\boldsymbol{x}|=k$ and $k\in\mathcal{K}$.
Consider the LPs with the individual constraints
\begin{equation}
\label{eq:appB1}
F_{\min/\max}^{\rm ind}
=\underset{\boldsymbol{p}\in\mathcal{P}_{2^n}}{\min/\max}
\left\{
p_{\boldsymbol{0}}:
\sum_{\boldsymbol{\lambda}\in\mathbb{Z}_2^n}
(-1)^{\boldsymbol{x}\cdot\boldsymbol{\lambda}}
p_{\boldsymbol{\lambda}}
=c_k,\ 
\forall\,\boldsymbol{x}\ \text{with }|\boldsymbol{x}|=k,\ k\in\mathcal{K}
\right\},
\end{equation}
or only the stabilizer-moment constraints
\begin{equation}
\label{eq:appB2}
F_{\min/\max}^{\rm mom}
=\underset{\boldsymbol{p}\in\mathcal{P}_{2^n}}{\min/\max}
\left\{
p_{\boldsymbol{0}}:
\frac{1}{L_k}
\sum_{\substack{\boldsymbol{x}\in\mathbb{Z}_2^n\\|\boldsymbol{x}|=k}}
\sum_{\boldsymbol{\lambda}\in\mathbb{Z}_2^n}
(-1)^{\boldsymbol{x}\cdot\boldsymbol{\lambda}}
p_{\boldsymbol{\lambda}}
=c_k,\ 
\forall\,k\in\mathcal{K}
\right\}.
\end{equation}
Then $F_{\min}^{\rm ind}=F_{\min}^{\rm mom}$ and $F_{\max}^{\rm ind}=F_{\max}^{\rm mom}$.
\end{lemma}

\begin{proof}
Straightforwardly, every distribution feasible for Eq.~\eqref{eq:appB1} is also feasible for Eq.~\eqref{eq:appB2}. Conversely, let $\boldsymbol{p}$ be feasible for Eq.~\eqref{eq:appB2} and symmetrize it as
\begin{equation}
\overline{\boldsymbol{p}}
=\frac{1}{n!}\sum_{\pi\in S_n}\pi\cdot\boldsymbol{p},
\qquad
(\pi\cdot\boldsymbol{p})_{\boldsymbol{\lambda}}
\coloneqq p_{\pi^{-1}\boldsymbol{\lambda}}.
\end{equation}
This symmetrization preserves $p_{\boldsymbol{0}}$. For every
$\boldsymbol{x}$ with $|\boldsymbol{x}|=k$, it also gives
\begin{equation}
\begin{aligned}
\sum_{\boldsymbol{\lambda}\in\mathbb{Z}_2^n}
(-1)^{\boldsymbol{x}\cdot\boldsymbol{\lambda}}
\overline{p}_{\boldsymbol{\lambda}}
=\frac{1}{n!}\sum_{\pi\in S_n}
\sum_{\boldsymbol{\lambda}\in\mathbb{Z}_2^n}
(-1)^{\boldsymbol{x}\cdot\boldsymbol{\lambda}}
p_{\pi^{-1}\boldsymbol{\lambda}} 
=\frac{1}{n!}\sum_{\pi\in S_n}
\sum_{\boldsymbol{\lambda}\in\mathbb{Z}_2^n}
(-1)^{(\pi^{-1}\boldsymbol{x})\cdot\boldsymbol{\lambda}}
p_{\boldsymbol{\lambda}} 
=\frac{1}{L_k}
\sum_{\substack{\boldsymbol{y}\in\mathbb{Z}_2^n\\|\boldsymbol{y}|=k}}
\sum_{\boldsymbol{\lambda}\in\mathbb{Z}_2^n}
(-1)^{\boldsymbol{y}\cdot\boldsymbol{\lambda}}
p_{\boldsymbol{\lambda}}
=c_k.
\end{aligned}
\end{equation}
Thus, $\overline{\boldsymbol{p}}$ is also feasible for Eq.~\eqref{eq:appB1} with the same objective value $p_{\boldsymbol{0}}$. The two feasible sets therefore attain the same set of objective values, so both their minima and their maxima coincide.
\end{proof}

\section{Numerical procedure for Fig.~\ref{fig:measurement_cost}}
\label{app:fig2-procedure}

We describe how each point in the main measurement-cost panels of Fig.~\ref{fig:measurement_cost} is obtained. For each system size $n$ and prescribed fidelity accuracy $\upDelta$, we fix $\delta=0.01$ and choose the noise strength to give a true fidelity of $0.6$. The LP centres $\tilde{\mu}_k$ are set to the expectation values predicted by the noise model, while the tolerance $\epsilon$ is retained in every measurement constraint to account for finite statistics. Figure~\ref{fig:fig2-workflow} summarizes the procedure for both noise models.

\begin{figure}[htbp]
    \centering
    \begin{tikzpicture}[
        font=\small,
        process/.style={draw=red!65!black,fill=red!2,rounded corners=2pt,
            align=center,text width=7.6cm,inner xsep=7pt,inner ysep=6pt},
        model/.style={process,text width=5.6cm,minimum height=2.05cm},
        decision/.style={diamond,aspect=3,draw=red!65!black,fill=red!5,
            align=center,inner xsep=5pt,inner ysep=4pt},
        flow/.style={-{Latex[length=2.2mm]},draw=red!65!black,semithick},
        branch/.style={font=\footnotesize,fill=white,inner sep=2pt}
    ]
    \node[process] (input) at (0,0) {
        Fix $n$, $\upDelta$, $\delta=0.01$, and true fidelity $F=0.6$\\
        Select the global or local noise model.
    };
    \node[model] (global) at (-3.7,-2.55) {
        \textbf{Global depolarizing noise}\\
        Set $\eta=0.4/(1-2^{-n})$\\
        and $\tilde{\mu}_k=1-\eta$.
    };
    \node[model] (local) at (3.7,-2.55) {
        \textbf{Local depolarizing noise}\\
        Choose $p$ with $2^{-n}G(1)=0.6$\\
        and compute $\tilde{\mu}_k$ by the transfer matrix\\
        (Appendix~\ref{sec:supp-local-depolarizing}).
    };
    \node[process] (setup) at (0,-5.15) {
        Choose $\epsilon\in\{10^{-2},10^{-3},10^{-4}\}$\\
        and start with $K=1$.
    };
    \node[process] (lp) at (0,-7.05) {
        \textbf{Solve the LP in Eq.~\eqref{eq:main}}\\
        Use $J_K=\{S^{(1)},\ldots,S^{(K)}\}$ with tolerance $\epsilon$\\
        to obtain $F_{\min}$ and $F_{\max}$.
    };
    \node[decision,text width=3.8cm] (width) at (0,-9.65) {
        $F_{\max}-F_{\min}\leq2\upDelta$?
    };
    \node[process] (cost) at (0,-12.25) {
        \textbf{Stop at $K_{\rm stop}=K$}\\
        Evaluate $l$, $M$, and the appropriate cost bounds at $K_{\rm stop}$.
    };
    \node[process] (plot) at (0,-14.25) {
        Repeat for the three tolerances and plot\\
        the smallest valid expected-cost upper bound.
    };

    \draw[flow] (input.south) -- (0,-1.0) -| (global.north);
    \draw[flow] (input.south) -- (0,-1.0) -| (local.north);
    \draw[flow,-] (global.south) |- (0,-4.05);
    \draw[flow,-] (local.south) |- (0,-4.05);
    \draw[flow] (0,-4.05) -- (setup.north);
    \draw[flow] (setup) -- (lp);
    \draw[flow] (lp) -- (width);
    \draw[flow] (width) -- node[branch,right] {Yes} (cost);
    \draw[flow] (cost) -- (plot);
    \draw[flow] (width.east) -- node[branch,above] {No ($K<n$)} (6.1,-9.65)
        -- node[branch,right] {$K\leftarrow K+1$} (6.1,-7.05) -- (lp.east);
    \end{tikzpicture}
    \caption{Numerical procedure for the main panels of Fig.~\ref{fig:measurement_cost}. The global branch applies to arbitrary hypergraph target states, while the local branch uses the 1D order-3 cluster state. If the width criterion is not met for any $K\leq n$, the corresponding tolerance is excluded from the cost minimization. The measurement allocations and cost bounds are given in Appendix~\ref{app:large-deviation bounds}.}
    \label{fig:fig2-workflow}
    \label{fig:fig2a-workflow}
\end{figure}
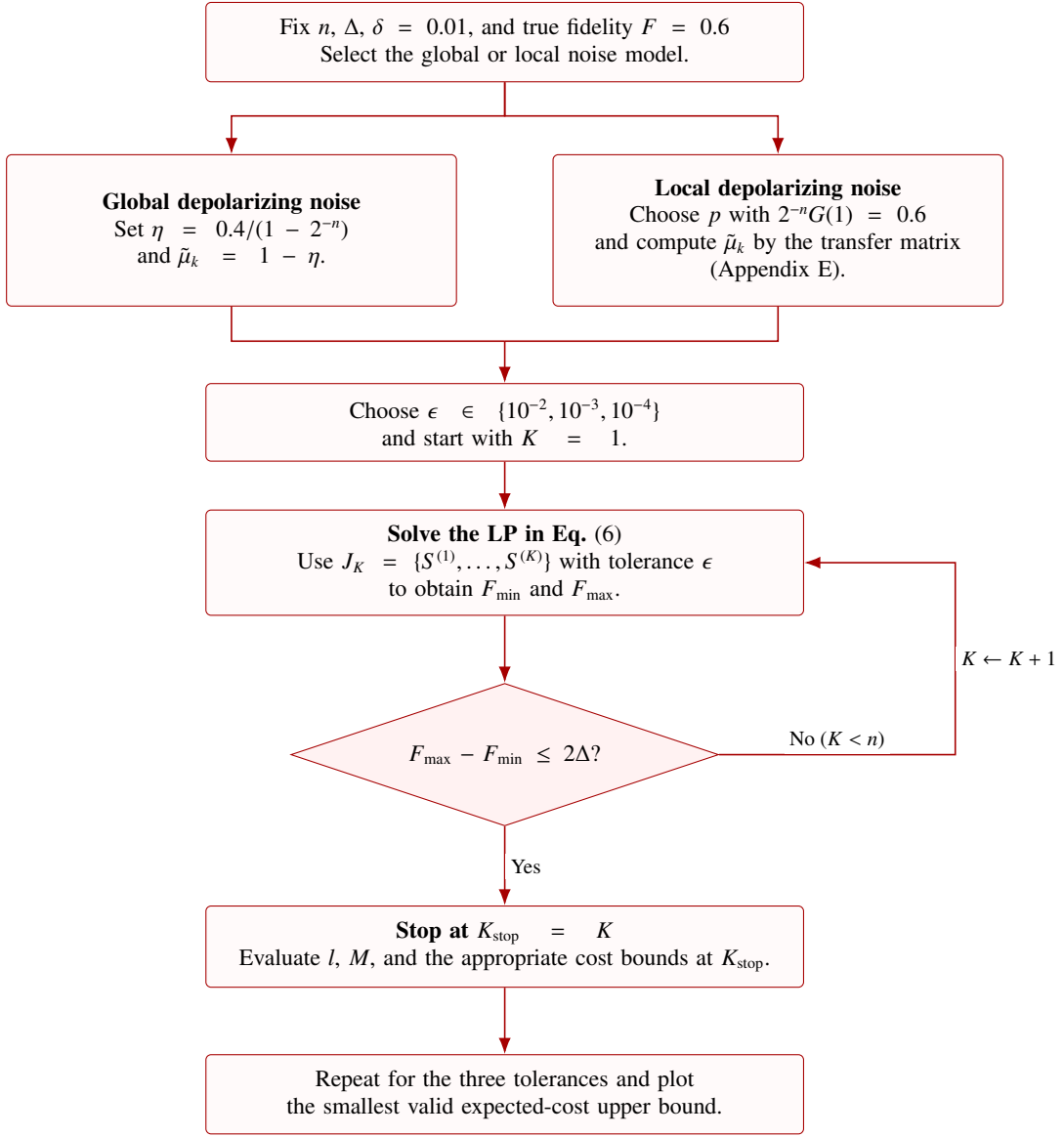

\end{document}